\documentclass[reqno]{amsart}

\usepackage{amsmath,amssymb,amsthm,amsfonts}
\usepackage{color}
\usepackage{hyperref}

\theoremstyle{plain}
\newtheorem{prop}{Proposition}[section]
\newtheorem{theorem}[prop] {Theorem}

\theoremstyle{definition}
\newtheorem{lemma}[prop]{Lemma}
\newtheorem{definition}[prop]{Definition}
\newtheorem{corollary}[prop]{Corollary}

\theoremstyle{remark}
\newtheorem*{remark}{Remark}
\newtheorem*{example}{Example}

\newcommand{\N}{\mathbb{N}}
\newcommand{\R}{\mathbb{R}}
\newcommand{\Z}{\mathbb{Z}}
\newcommand{\C}{\mathbb{C}}
\renewcommand{\P}{\mathbb{P}}

\newcommand{\dd}{\mathrm{d}} %integration d
\newcommand{\eps}{\epsilon}

\renewcommand{\eps}{\varepsilon}
\renewcommand{\phi}{\varphi}

\newcommand{\vect}[1]{\boldsymbol{#1}}

\DeclareMathOperator{\supp}{supp}

\title{Cluster and virial expansions for the multi-species Tonks gas}

\author{S. Jansen}
\thanks{\emph{Address}: Ruhr-Universit{\"a}t Bochum, Fakult{\"a}t f{\"u}r Mathematik, 44780 Bochum, Germany.\\
\emph{Email}: sabine.jansen@rub.de}
\date{9 March 2015}

\begin{document}

\begin{abstract}
	We consider a mixture of non-overlapping rods of different lengths $\ell_k$ moving in $\mathbb{R}$ or $\mathbb{Z}$.
        Our main result are necessary and sufficient convergence criteria for the expansion of the pressure in terms of the activities $z_k$ and the densities $\rho_k$.  This provides an explicit example against which to test known cluster expansion criteria, and illustrates that for non-negative interactions, the virial expansion can converge in a domain much larger than the activity expansion. In addition, we give explicit formulas that generalize the well-known relation between non-overlapping rods and labelled rooted trees. We also prove that for certain choices of the activities, the system can undergo a condensation transition akin to that of the zero-range process. The key tool is a fixed point equation for the pressure. \\

\noindent \emph{Keywords:} Cluster and virial expansions, non-overlapping rods, combinatorics of labelled colored trees, close-packing transition.
\end{abstract}
\maketitle

%%%%%%%%%%%%%%%%%%%%%%%%%%%%%%%%%%%%%%%%%%%%%%%%%%%%%%%%%%%%%%%%%%%%%%%55

\section{Introduction} 

The present article deals with one-dimensional systems of non-overlapping rods on a continuous segment $[0,L]$  or a discrete interval $\{0,1,\ldots,L-1\}$. There are countably many types $k$ of rods, coming each with a length $\ell_k\geq 0$ and an activity $z_k$. A rich literature deals with related models: our model is a multi-species variant of the well-known \emph{Tonks gas} \cite{tonks36}.  We may also view it as a one-dimensional special case of hard spheres mixtures \cite{lebowitz-rowlinson64}.  A good control of discrete one-dimensional partition functions enters as a building block for two-dimensional models with orientational long range order~\cite{ioffe-velenik-zahradnik06, disertori-giuliani13}. The one-dimensional discrete system of rods also appears in stationary distributions for driven one-dimensional systems~\cite{kafri02}, which in turn are closely related to the zero-range process where particles are piled up rather than aligned in a rod \cite{evans-hanney05}. Phase transitions for one-dimensional cluster models have been studied in detail by Fisher and Felderhof~\cite{fisher-felderhof70,fisher72}. 

Our principal motivation comes from the model's solvability and the specific, though model-dependent, answers to questions on cluster expansions it allows. The first question concerns domains of convergence. There are many sufficient convergence criteria available, but it is an ongoing effort to improve them, see for example~\cite{fernandez-procacci, bissacot-fernandez-procacci}. This raises the question of how much room for improvement there actually is. We answer this question for the multi-species Tonks gas by determining the exact domain of convergence (see Eqs.~\eqref{eq:continuous} and~\eqref{eq:discrete} below). The answer for general models can of course be quite different, but we hope that our results will serve as a helpful control group in future studies. 

The second question is how the domain of convergence of the activity expansion compares to the domain of analyticity. It is common wisdom that for repulsive interactions, the activity expansion ceases to converge before a phase transition occurs~\cite{ruelle-book}. For single-species model this means that the radius of convergence is strictly smaller than the activity value at which a phase transition occurs, if it occurs at all. We prove that for the multi-species Tonks gas with rod lengths $\ell_k=k$, the  difference between convergence and analyticity domain is even more drastic: for the convergence of the activity expansion it is necessary that the activities go to zero exponentially fast $z_k= O(\exp(-ak))\to 0$, while the pressure stays analytic all the way up to exponentially diverging activities $z_k \to \infty$ (Corollary~\ref{cor:domains}).

The third question concerns the virial expansion, i.e., the expansion of the pressure in terms of the densities $\rho_k$. It has been suggested that the virial expansion can be more advantageous than the activity expansion~\cite{brydges-iamp}, and indeed for some models this is known to be true ~\cite{joyce88,brydges-marchetti14}. 
 We prove that the same holds for the multi-species Tonks gas, again in a quite drastic way: the virial expansion converges in all of the analyticity region (Theorems~\ref{thm:vir} and~\ref{thm:vir-discrete}, Corollary~\ref{cor:domains}), including densities $\rho_k(z_1,z_2,\ldots)$ that correspond to exponentially diverging activities $z_k \to \infty$ (when $\ell_k =k$). 

In addition, we provide explicit formulas for the pressure-activity expansion that are interesting from a combinatorics point of view (Theorems~\ref{thm:cluster-convergence} and~\ref{thm:convergence-discrete}). In the continuous case, we find that the activity expansion is (up to signs) the multivariate exponential generating function for labelled rooted colored trees. The vertex colors correspond to rod types and the trees have weights that depend on the lengths $\ell_k$. This generalizes the well-known relationship between the exponential generating function $T(z) = \sum_{n\geq 1} z^n n^{n-1}/n!$ of rooted labelled trees and the pressure for non-overlapping rods of length $1$ (see~\cite{brydges-imbrie} and the references therein). It would be interesting to know whether the answers to corresponding combinatorial puzzles given in~\cite{bernardi08,tate14} extend to the multi-species setting. 

Let us describe in more detail our results on the convergence domain of the activity expansion of the pressure in the infinite volume limit. In the continuous case, the expansion converges absolutely if and only if (Theorem~\ref{thm:cluster-convergence})\footnote{The criterion~\eqref{eq:continuous} refers to the pressure-activity expansion $p(\vect{z})$. The convergence of the density-activity expansion $\rho_k(\vect{z})$ in general requires a strict inequality $\sum_k|z_k|\exp( a\ell_k) <a$. The same remark applies to the discrete system.}
\begin{equation}  \label{eq:continuous}
  \exists a>0: \quad \sum_{k=1}^\infty |z_k|e^{a \ell_k} \leq a.
\end{equation} 
In the discrete case we may choose $\ell_k =k$, and the cluster expansion for the pressure  converges if and only if (Theorem~\ref{thm:convergence-discrete})
\begin{equation}  \label{eq:discrete}
  \exists a>0: \quad \sum_{k=1}^\infty |z_k| e^{ak} \leq e^a-1.  
\end{equation}
The principal novelty lies in the \emph{only if} part: if the criterion fails, then the expansion is not absolutely convergent. Eqs.~\eqref{eq:continuous} and~\eqref{eq:discrete} should be compared with the following known \emph{sufficient} criteria. In the continuous case, it is enough that for some $a>0$ and all $k\in \N$ \cite{poghosyan-ueltschi09}
\begin{equation} \label{eq:kp} 
  \sum_{j=1}^\infty \int_{-\infty}^\infty| z_j| e^{a \ell_j} \mathbf{1}\Bigl( [x,x+\ell_j]\cap [ 0,\ell_k]\neq \emptyset \Bigr) \dd x = \sum_{j=1}^\infty (\ell_j + \ell_k) z_j e^{a \ell_j} \leq a \ell_k. 
\end{equation} 
In the discrete case, it is enough that for some $a>0$ \cite{gruber-kunz71,fernandez-procacci} 
\begin{equation}\label{eq:gruber-kunz}
  \sum_{x\in \Z,\ k \in \N} |z_k| \mathbf{1}\bigl( \{x,\ldots,x+k-1\} \ni 0) e^{ak} = \sum_{k=1}^\infty k e^{ak}|z_k| \leq e^{a} -1. 
\end{equation} 
The activity domains determined by the sufficient criteria~\eqref{eq:kp} and~\eqref{eq:gruber-kunz} are clearly smaller than the full convergence domains given by \eqref{eq:continuous} and~\eqref{eq:discrete}, however we shall take the point of view that the difference is small: in the discrete case both~\eqref{eq:discrete} and~\eqref{eq:gruber-kunz} are of the form $||\boldsymbol{z}||_a \leq \exp(a)-1$   with weighted norms $||\boldsymbol{z}||_a$ that impose the same exponential decay of $z_k$ and differ merely by the prefactor $k$. From this point of view the one-dimensional model of non-overlapping rods provides an example for which the classical convergence criteria are already nearly optimal.\footnote{This point of view focuses on \emph{qualitative} features of the domain of convergence for objects of \emph{unbounded} size ($\ell_k\to \infty$). A different question is about \emph{quantitative} estimates on the radius of convergence for hard-sphere systems of \emph{bounded} size, see e.g.  \cite{fernandez-procacci-scoppola}.}

A full list of results is given in the next section. In addition to  exact formulas for the activity and density expansions as well as their domains of convergence, we prove that the pressure solves a fixed point equation. In the continuous case it reads 
\begin{equation} \label{eq:pressure-equation}  
	p(\vect{z}) = \sum_k z_k \exp\Bigl(- \ell_k p(\vect{z}) \Bigr)
\end{equation} 
and is satisfied by the pressure whenever the equation has a solution (Theorem~\ref{thm:pressure}).
In the discrete case the equation is instead (Theorem~\ref{thm:pressure-discrete})
\begin{equation} \label{eq:fixed-point-discrete} 
    1- \exp\bigl(- p(\vect{z})\bigr)= \sum_k z_k \exp( - k p(\vect{z})).
\end{equation}
(remember $\ell_k = k$).
Situations where the fixed point equation has no solution are possible and correspond to a close-packing regime (Theorem~\ref{thm:pack-frac}). Section~\ref{sec:explanations} provides two different explanations of the fixed point equation, a statistical mechanics explanation and a probabilistic explanation in terms of renewal equations. 

The fixed point equation is not only of interest in itself but also lies at the heart of all of our results and their proofs; it is no coincidence that Eqs.~\eqref{eq:pressure-equation} and~\eqref{eq:fixed-point-discrete} resemble so closely the convergence criteria~\eqref{eq:continuous} and~\eqref{eq:discrete} as well as functional equations for combinatorial generating functions. This will become clear in the proofs, given in Sections~\ref{sec:fixed-point} to~\ref{sec:virial}. Section~\ref{sec:inverse} proves a technical result (Theorem~\ref{thm:inverse}) on the inversion of the density-activity relation via the inverse function theorem.

%%%%%%%%%%%%%%%%%%%%%%%%%%%%%%%%%%%%%%%%%%%%%%%%%%%%%%%%%%%%%5

\section{Model and results} \label{sec:results}

\subsection{Continuous system}

Let $(\ell_k)_{k\in \N}$ be a family of rod lengths $\ell_k \geq 0$ and $\vect{z} = (z_k)_{k\in \N}$ positive activities $z_k \geq 0$. The rod lengths do not need to go to infinity; in fact the case $\ell_k \to 0$ might be of interest in view of the transition towards dust studied in some fragmentation models~\cite{haas04}. Let $\mathcal{I}\subset \N_0^\N$ be the set of multi-indices $(N_k)_{k\in \N}$   that have none or finitely many non-vanishing entries, and $\mathcal{I}^*:= \mathcal{I}\backslash \{\vect{0}\}$ the set of multi-indices with at least one non-zero entry. 
For $\vect{n} \in \mathcal{I}$ we set
$\vect{n}!:= \prod_k (n_k!)$ and $\vect{z}^{\vect{n}}:=\prod_k z_k^{n_k}$,
with the  convention $0^0=1$ and $0!=1$. The grand-canonical partition function for the volume  $[0,L]$ is 
\begin{multline} \label{eq:partfct-cont}
	\Xi_L(\vect{z}) := 1+ \sum_{\boldsymbol{N} \in \mathcal{I}^*} \frac{\vect{z}^{\vect{N}}}{\vect{N}!} 
		\int_{[0,L]^{\sum_k N_k}}\dd x_{11} \cdots \dd x_{1N_1} \dd x_{21} \cdots \dd x_{2N_2} \cdots \\ \times \mathbf{1}\Bigl( \forall k,j:\ [x_{kj}, x_{kj} + \ell_k]\subset [0,L] \Bigr) \\
\times  \mathbf{1} \Bigl( \forall (k,j) \neq (m,i):\ [x_{kj}, x_{kj} + \ell_k] \cap  [x_{mi}, x_{mi} + \ell_m]= \emptyset \Bigr).
\end{multline}
In the integral $x_{kj}$ is the left end point of a rod of length $\ell_k$. The indicators ensure that rods do not overlap and lie entirely in $[0,L]$. The pressure is 
\begin{equation} \label{eq:pressure-def}
	p(\vect{z}):= \lim_{L\to \infty}\frac{1}{L} \log \Xi_L(\vect{z})
\end{equation} 
The existence of the limit in $[0,\infty) \cup\{\infty\}$  follows from general subadditivity arguments, a more  concrete determination is given in Theorem~\ref{thm:pressure} below. We shall see that the stability condition 
\begin{equation} \label{eq:stability} 
	\exists \theta \in \R:\quad \sum_k z_k\exp(\theta \ell_k) < \infty
\end{equation}
ensures that the limit defining $p(\vect{z})$ is finite. Note that we allow for $\theta<0$; in particular, for $\ell_k =k$, the activities are allowed to diverge as $\exp(ck)$.  
Let $\theta^*$  
\begin{equation} \label{eq:thetastar}
	\theta^*:= \sup\{\theta \in \R\mid  \sum_k z_k\exp(\theta \ell_k) < \infty\}
\end{equation}   
be the abscissa of convergence of  the Dirichlet type series 
\begin{equation} \label{eq:g}
     g(\theta) = \sum_k z_k \exp(\theta \ell_k).
\end{equation} 
For notational convenience we suppress the $\vect{z}$-dependence in $\theta^*$ and $g(\theta)$. Depending on the values of $\vect{z}$, the fixed point equation $g(\theta) = - \theta$ may or may not have a solution. As we shall see, the cases correspond to different physical behaviors and the following names are convenient.  

\begin{definition} \label{def:regimes}
  Let $\vect{z} = (z_k)_{k\in\N}$ be non-negative activities that satisfy the stability condition~\eqref{eq:stability}. Then $\vect{z}$ belongs to the 
  \begin{enumerate}
    \item [(a)] \emph{fluid domain} $\mathcal{D}_{\rm fluid}$ if $\theta^*=\infty$ or $\theta^*<\infty$ and $g(\theta^*)>-\theta^*$;
      \item [(b)] \emph{close-packing domain}  if $\theta^*<\infty$ and $g(\theta^*)<-\theta^*$. 
       \item [(c)] \emph{transition domain} if $g(\theta^*) = - \theta^*$.
 \end{enumerate}
\end{definition}
In the fluid domain the fixed point equation $g(\theta) = - \theta$ has a unique solution $\theta <\theta^*$, in the close-packing domain it has no solution, and in the transitional domain the unique solution is $\theta=\theta^*$. 

\begin{theorem}\label{thm:pressure} 
      Let $(z_k)_{k\in \N}$ be non-negative activities satisfying the stability condition~\eqref{eq:stability}. Then
	\begin{enumerate} 
		\item[(a)] In the fluid domain the pressure $p(\vect{z})$ is given by the unique solution $p>-\theta^*$ of 
			$\sum_k z_k \exp( - p \ell_k ) = p$ .
		\item [(b)] In the close-packing and transition domains the pressure is  $p(\vect{z}) =  -\theta^*$.  
	\end{enumerate}
\end{theorem} 
The theorem is proven in Section~\ref{sec:fixed-point}, an explanation of the fixed point equation is given in Section~\ref{sec:explanations}. Theorem~\ref{thm:pressure} suggests the possibility of phase transitions, illustrated by the following example. 

\begin{example}
	Take rod lengths $\ell_k = k$ and parameter-dependent activities $z_k(\mu) =\exp(k \mu - \sqrt{k})$, $\mu\in \R$. The activities are associated with a parameter-dependent Dirichlet series $ \sum_k \exp[(\theta +\mu)k - \sqrt{k}]$. The abscissa of convergence is $\theta^*= - \mu$  and at $-\theta^*$ the Dirichlet series takes the $\mu$-independent value $\sum_k \exp( - \sqrt{k})$. In the fluid regime  $\mu< \sum_k \exp( - \sqrt{k})$ the pressure $p(\mu)$ is the unique solution to $p = \sum_k \exp( k \mu - k p  - \sqrt{k})$, and an implicit function theorem shows that $p(\mu)$ is analytic with 
\begin{equation*} 
  \frac{\dd p}{\dd \mu} = \frac{\sum_k k \exp( k [\mu - p] - \sqrt{k})}{1+\sum_k k \exp( k [\mu - p] - \sqrt{k})} \leq \frac{\sum_k k\exp( -\sqrt{k})}{1+ \sum_k k\exp( -\sqrt{k})} <1.
\end{equation*} 
In the close-packing regime $\mu>\sum_k \exp( - \sqrt{k})$, the pressure is $p(\mu)=\mu$, and we recognize a first-order phase transition at $\mu = \sum_k \exp(  - \sqrt{k})$. 
\end{example} 

\begin{remark} 
  The previous example is easily extended.
  Let $(a_k)_{k\in \N}$ be non-negative weights such that $h(z)=\sum_k a_k z^k$ has positive radius of convergence $R>0$. Set $z_k(\mu):= a_k \exp(k\mu)$. Then there is a phase transition in $\mu$ if and only if $R$ is finite and $\sum_k a_k R^k<\infty$. The transition takes place at $\exp(\mu) = R$. If $\sum_k ka_k R^k<\infty$, it is of first order. The situation is closely related to condensation in the zero-range process~\cite{evans-hanney05}, with the important difference, however, that we can take $\exp(\mu)>R$ and go beyond the coexistence region, into the close-packed phase. The phenomenon is also similar to phase transitions studied in Fisher-Felderhof clusters~\cite{fisher-felderhof70,fisher72}. 
\end{remark}

A very heuristic explanation of  the phase transition is the following. 
Suppose that $\ell_k = k$ and the abscissa of convergence $\theta^*=- \limsup_{k\to \infty} \frac{1}{k} \log z_k$ is finite. Then $\exp( - L \theta^*)$ is the weight of the configuration where space is filled with one long rod. In order to see how much we loose by breaking space into smaller rods, it seems natural to rescale the activities as $z_k \to z_k \exp( k \theta^*)$. If we neglect excluded volumes, the partition function for a system of many rods  for the rescaled activities becomes $\exp( L\sum_{k=1}^\infty z_k \exp( k \theta^*)) = \exp( L g(\theta^*))$.
When $g(\theta^*)<-\theta^*$, it seems more advantageous to fill space with one long rod, and so we recover the case distinction from Theorem~\ref{thm:pressure}, though the use of rescaled activities in this argument is somewhat ad hoc. 

The next theorem shows that the fluid and close-packing domains do indeed correspond to different behaviors of the system. The relevant order parameter is the \emph{packing fraction}, the fraction of volume covered by rods. The convergence $\sum_k \ell_k N_k/L\to \sigma$ stands for convergence in probability in the grand-canonical ensemble, i.e., 
 for all $\eps>0$, the grand-canonical probability that $|\sum_k N_k \ell_k /L- \sigma |\geq \eps$ goes to $0$ as $L\to \infty$. 

\begin{theorem} \label{thm:pack-frac} 
  Let $(z_k)_{k\in \N}$ be non-negative activities satisfying the stability assumption~\eqref{eq:stability}. Then as $L\to \infty$ the packing fraction behaves as follows:
	\begin{enumerate} 
		\item[(a)] In the fluid regime 
			\begin{equation*} 
				\frac{1}{L} \sum_k \ell_k N_k \to \sigma(\vect{z})<1,\quad 
                                \sigma(\vect{z}) = \frac{\sum_k \ell_k z_k \exp[-\ell_k p(\vect{z})]}{1+ \sum_k \ell_k z_k \exp[-\ell_k p(\vect{z})]}.
		\end{equation*} 
		\item[(b)] In the close-packing regime, the packing fraction converges to $1$. 
	\end{enumerate} 
\end{theorem} 
The theorem is proven in Section~\ref{sec:fixed-point} by a large deviations approach.

\begin{remark} 
  In the transition regime $g(\theta^*) = -\theta^*$ there are two possible scenarios: if as $\theta$ approaches the abscissa of convergence the derivative $g'(\theta) = \sum_k \ell_k \exp( \theta \ell_k)$ diverges, the packing fraction converges to $1$. If the derivative stays bounded, let $\sigma^*:=\lim_{\theta \nearrow \theta^*} g'(\theta)/[1+ g'(\theta)]$.  Then $\sigma^*\in (0,1)$ and Lemmas~\ref{lem:pvar} and~\ref{lem:ldp} below only show that the grand-canonical probability of seeing a packing fraction smaller than $\sigma^*-\eps$ goes to zero, for every $\eps>0$. Intuitively, this corresponds to a coexistence region between a densely packed phase ($\sigma = 1$) and a fluid  that saturates at $\sigma = \sigma^*$. 
\end{remark}

Now we give the domain of convergence of the cluster expansion, as discussed in the introduction, and an explicit formula for the expansion. The formula generalizes a well-known relationship between a tree generating function and the pressure for non-overlapping rods of length $1$, see~\cite{brydges-imbrie} and the references therein.

\begin{theorem} \label{thm:cluster-convergence} 
  Let $(z_k)_{k\in \N}$ be non-negative activities satisfying the stability assumption~\eqref{eq:stability}. The following holds:  
	\begin{enumerate} 
		\item[(a)] If $\sum_k z_k \exp(a \ell_k) \leq a$ for some $a>0$, then the pressure defined by Eq.~\eqref{eq:pressure-def} is given by
                   \begin{equation} \label{eq:tonks} 
	           p(\vect{z}) = \sum_{\vect{n} \in \mathcal{I}^*} \frac{\vect{z} ^{\vect{n}} }{\vect{n}!}
(- \sum_k n_k\ell_k )^{\sum_k n_k -1} 
	\end{equation}
            and the sum is absolutely convergent. 
	\item[(b)] If $\sum_k z_k \exp(a \ell_k) >a$ for all $a>0$, then
                   \begin{equation*}
                       \sum_{\vect{n} \in \mathcal{I}^*}\Bigl| \frac{\vect{z} ^{\vect{n}} }{\vect{n}!}
(- \sum_k n_k\ell_k )^{\sum_k n_k -1} \Bigr| =\infty. 
                   \end{equation*}
	\end{enumerate}  
\end{theorem} 
The formula~\eqref{eq:tonks} is proven in Section~\ref{sec:trees}, where we also provide a combinatorial interpretation in terms of tree generating functions.  The convergence is addressed in Section~\ref{sec:convergence}. We should stress  that (a) refers only to the convergence of the expansion of the pressure; for the convergence of the expansions of the densities $\rho_k(\vect{z})$, the inequality~\eqref{eq:continuous} in general has to be strict for some $a>0$.

\begin{remark} 
  For negative (or complex) activities  with $\sum_k |z_k|\exp( a \ell_k) \leq a$ for some $a>0$, the sum~\eqref{eq:tonks} is absolutely convergent as well and we adopt Eq.~\eqref{eq:tonks} as a definition of the pressure. In Section~\ref{sec:convergence} we will see that $F(\vect{z}) = - p( - \vect{z})$ satisfies the flipped fixed point equation $F= \sum_k z_k \exp( \ell_k F)$; in fact the convergence condition~\eqref{eq:continuous} is equivalent to the existence of a solution to the flipped equation. This is interesting because $F$ is the pressure for a three-dimensional multi-type branched polymer~\cite{brydges-imbrie03b}. 
\end{remark} 

Finally we examine the virial expansion, i.e., the expansion in terms of the densities 
\begin{equation*} 
 \rho_k(\vect{z}):= z_k\frac{\partial p}{\partial z_k} (\vect{z}).
\end{equation*}

\begin{theorem}\label{thm:vir}
  Let $\vect{z}=(z_k)_{k\in \N}$ be non-negative activities in the fluid domain. Then the partial derivatives $\frac{\partial p}{\partial z_k}(\vect{z})$, $k\in \N$, at $\vect{z}$ exist and we have 
	\begin{equation} \label{eq:vir}
		p(\vect{z})  = \frac{\sum_k \rho_k (\vect{z})}{1- \sum_k \ell_k \rho_k (\vect{z})},  \quad 
		z_k = \frac{\rho_k(\vect{z})\exp( \ell_k p(\vect{z}))}{1- \sum_j \ell_j \rho_j(\vect{z})}
	\end{equation} 
	with $  \sum_k \ell_k \rho_k(\vect{z}) ) = \sigma(\vect{z}) < 1$  as in Theorem~\ref{thm:pack-frac}(a). 
\end{theorem}

Theorem~\ref{thm:vir} is proven in Section~\ref{sec:virial} by applying an implicit function theorem  to the fixed point equation~\eqref{eq:pressure-equation}. This approach also shows that the pressure is analytic in all of $\mathcal{D}_{\rm fluid}$. 

\begin{corollary} \label{cor:analyticity}
  Fix non-negative activities $\vect{z}^0 =(z_k^0)_{k\in \N}\in \mathcal{D}_{\rm fluid}$. Then for every $k\in \N$, the map $z_k \mapsto p(z_1^0,\ldots, z_{k-1}^0,z_k,z_k^0,\ldots)$ is analytic in some neighborhood of $z_k^0$. 
\end{corollary}
A similar statement holds for suitable parameter-dependent activities. In particular, the parametrization  $z_k(t) = t z_k$ allows us to connect vanishing activities to every $\vect{z} \in \mathcal{D}_{\rm fluid}$ in such a way that the pressure is an analytic function of the parameter $t$.  

\begin{proof} 
  For simplicity we consider $k=1$, the other cases are similar. Let $\theta^*$ be the abscissa of convergence of $\sum_k z_k^0 \exp( \ell_k \theta)$ and $f(z_1,\theta):= \theta + z_1 \exp(\theta \ell_1) + \sum_{k \geq 2} z_k^0 \exp(\theta \ell_k )$. Then $f$ is a holomorphic function of  $(z_1,\theta)$ in $\C\times \{\theta \mid {\rm Re}\, \theta <\theta^*\}$. We have $f(z_1^0,- p(\vect{z}^0)) = 0$ and 
\begin{equation*} 
  \frac{\partial f}{\partial \theta}\bigl(z_1^0,- p(\vect{z}^0)\bigr) = 1 + \sum_{k\geq 1} \ell_k z_k^0 e^{- \ell_k p(\vect{z}^0)} = \frac{1}{1 - \sigma(\vect{z}^0)} \neq 0.
\end{equation*} 
The holomorphic implicit function theorem~\cite[Chapter 7]{fritzsche-grauert-book} (see also \cite{sokal09} for a power series approach) guarantees the existence of a holomorphic function $\theta(z_1)$ defined in some open complex neighborhood of $0$  such that $\theta(z_1^0) = - p(\vect{z}^0)$ and $f(\theta(z_1),z_1) =0$. When $z_1$ is real and close enough to $z_1^0$, then $\vect{z}=(z_1,z_2^0,z_3^0,\ldots)$ must be in $\mathcal{D}_{\rm fluid}$. Indeed, changing only one coefficient $z_1$ does not change the abscissa of convergence $\theta^*$, and the inequality $\sum_k z_k^0 \exp( \ell_k \theta^*) > - \theta^*$ holds by continuity when $z_1^0$ is replaced by $z_1$ sufficiently close to $\theta^*$. It follows that $p(\vect{z})$ solves the fixed point equation and therefore $z_1\mapsto p(z_1,z_2^0,z_3^0\ldots) = - \theta(z_1)$ is analytic.
\end{proof}

It follows that the domain of convergence of the activity expansion is in general smaller than the domain of analyticity $\mathcal{D}_{\rm fluid}$. In contrast, the pressure-density expansion converges absolutely when $\sum_k \ell_k \rho_k(\vect{z}) <1$, which is the case in all of $\mathcal{D}_{\rm fluid}$. 
As mentioned in the introduction, the difference between the domains is quite drastic: when $\ell_k =k$, the fluid domain and the activities for which the virial expansion converges can include diverging activities $z_k \to \infty$, while the convergence of the activity expansion requires exponentially decreasing activities $z_k = O(\exp( - a k))$. Precisely, let 
\begin{align*}
\mathcal{D}_{\mathrm{May}} &:=\{\vect{z} \in \R_+^\N \mid\exists a>0:\, \sum_k z_k\exp(ka ) \leq a\}\\
 \mathcal{D}_\mathrm{vir} &:= 
\{\vect{z} \in \R_+^\N \mid \vect{z} \in \mathcal{D}_{\rm fluid},\ \rho_k(\vect{z})\text{ exists for all $k$, and }\sum_k k \rho_k(\vect{z})<1 \} 
\end{align*} 
be the domains of convergence for the activity and the density expansions when $\ell_k = k$. 

\begin{corollary} \label{cor:domains}
	Assume $\ell_k = k$ for all $k \in \N$. Then 
        \begin{equation*} 
          \mathcal{D}_{\rm May} \subsetneqq \mathcal{D}_{\rm vir} = \mathcal{D}_{\rm fluid}. 
         \end{equation*} 
\end{corollary} 

\begin{proof}
  Let $\rho_k = c/k^3$ with $c>0$ small enough so that $\sum_k k \rho_k <1$. Let $p:= (\sum_k \rho_k)/(1 - \sum_k k \rho_k)$ and $z_k:= \rho_k \exp( k p)/(1- \sum_j j\rho_j)$. Then $p = \sum_k z_k \exp( - k p)$ and it follows that $p = p(\vect{z})$ and $\rho_k = \rho_k(\boldsymbol{z})$. Since $\sum_k k \rho_k<1$, the virial expansion converges and $\boldsymbol{z}$ is in $\mathcal{D}_{\rm vir}$. At the same time $z_k$ diverges as $k\to \infty$ exponentially fast, so the activity expansion cannot converge and $\vect{z} \notin \mathcal{D}_{\rm May}$. Thus $\mathcal{D}_{\rm May} \subsetneqq \mathcal{D}_{\rm vir}$. The identity $\mathcal{D}_{\rm vir} = \mathcal{D}_{\rm fluid}$ is part of Theorem~\ref{thm:vir}. 
\end{proof}

We conclude with a technical result concerning the use of inverse function theorems which complements Section 2.2 in~\cite{multivirial}. 
Often virial expansions are obtained by inverting the density-activity relation with the help of an inverse function theorem. This works for finitely many species because $\partial \rho_k /\partial z_j$ at $\vect{z}=0$ is the identity matrix and in particular invertible.
For infinitely many species, it is still true that the matrix $(\partial \rho_k/\partial z_j)$ of directional (G{\^a}teaux) derivatives at $\vect{z}=\vect{0}$ is the identity matrix. However the existence of directional derivatives is no longer enough to guarantee Fr{\'e}chet differentiability in suitable Banach spaces.  In theory one could imagine that this is just a technicality requiring additional estimates. The next theorem shows that to the contrary, for the natural Banach spaces at hand, the usual inversion procedure cannot work because $\vect{z}\mapsto \vect{\rho}(\vect{z})$ is not a bijection between neighborhoods of the origin. 

We restrict to $\ell_k =k$. For $a\geq 0$, let 	$||\vect{x}||_a:= \sum_{n=1}^\infty  |x_n|\exp(an)$ and 
$E_a:= \{ (x_n)_{n\in \N} \in \R^\N \mid ||x||_a < \infty\}$. These Banach spaces are natural because of the convergence criterion $||\vect{z}||_a \leq a$ from Theorem~\ref{thm:cluster-convergence}. They allow for negative and complex activities and densities.

\begin{theorem} \label{thm:inverse} 
  Suppose that $\ell_k = k$ for all $k\in \N$, and let $E_a,V_b\subset \C^\N$ be the spaces of complex activities and densities introduced above. There is no way to choose 
$a>0$, $b\geq 0$ and neighborhoods $U_a\subset E_a$ and $V_b\subset E_b$ of the origin so that $\vect{\rho}(\cdot)$ is a bijection from $U_a$ onto $V_b$. 
\end{theorem} 

Theorem~\ref{thm:inverse} is proven in Section~\ref{sec:inverse}. It explains why the virial expansion is much more delicate for infinitely many species than for finitely many species. A general result on multi-species virial expansions was nevertheless proven in~\cite{multivirial} using Lagrange-Good inversion~\cite{good60,ehrenborg-mendez}.

\subsection{Discrete system}
On the lattice we may take without loss of generality $\ell_k = k$, $k \in \N$. We use the multi-index notation from the previous section. The grand-canonical partition function for rods on a line $\{0,1,\ldots, L-1\}$ with $L\in \N$ is defined as in Eq.~\eqref{eq:partfct-cont}, except that the integration over rod end points is replaced by a summation over end points in $\{0,\ldots, L-1\}$, and the indicators should ensure that rods are contained in $\{0,1,\ldots,L-1\}$. The definition of the pressure, the Dirichlet type series $g(\theta)$,  the stability condition~\eqref{eq:stability}, and the abscissa of convergence~\eqref{eq:thetastar} are unchanged. The fixed point equation changes, however, and we define 
\begin{equation*}
  f(\theta) = - \log \bigl( 1- g(\theta)\bigr) = - \log \Bigl( 1- \sum_{k} z_k \exp( k \theta) \Bigr). 
\end{equation*} 
By expanding the logarithm, we find that $f(\theta)$ is a power series in $\exp(\theta)$ with positive coefficients and radius of convergence $R \leq \exp(\theta^*)$.
 The equation $f(\theta) = - \theta$ has a solution $\theta \leq \log R$ if and only if $g(\theta^*) \geq 1- \exp(\theta^*)$. 

\begin{definition} 
    Let $(z_k)_{k\in \N}$ be non-negative activities satisfying the stability condition~\eqref{eq:stability}. We say that the discrete system of rods is in the 
   \begin{enumerate}
     \item [(a)] fluid regime if $g(\theta^*) > 1- \exp(\theta^*)$;
     \item [(b)] close-packing regime if $g(\theta^*) < 1- \exp(\theta^*)$;
     \item [(c)] transition regime if $g(\theta^*) = 1- \exp(\theta^*)$.
   \end{enumerate}
\end{definition}

\begin{theorem} \label{thm:pressure-discrete}
  Let $(z_k)_{k\in \N}$ be non-negative activities satisfying the stability condition~\eqref{eq:stability}. Then 
  \begin{enumerate}
     \item [(a)] in the fluid regime the pressure $p(\vect{z})$ is the unique solution of $f(-p) = p$, and we have $p(\vect{z})>- \theta^*$.
     \item [(b)] In the close-packing and transition regimes the pressure is $p(\vect{z}) = - \theta^*$. 
   \end{enumerate}
\end{theorem}
The theorem is proven in Section~\ref{sec:fixed-point}.

%We leave it to the reader to formulate an analogue of Theorem~\ref{thm:pack-frac} on the behaviors of the packing fraction in the fluid and close-packing regimes. 

\begin{theorem} \label{thm:convergence-discrete}
  Let $(z_k)_{k\in \N}$ be non-negative activities satisfying the stability condition~\eqref{eq:stability}. The following holds:  
	\begin{enumerate} 
		\item[(a)] If $\sum_k z_k \exp(a k) \leq \exp(a) -1$ for some $a>0$, then the pressure defined by Eq.~\eqref{eq:pressure-def} is given by
                   \begin{equation} 
	           p(\vect{z}) = \sum_{\vect{n} \in \mathcal{I}^*} \frac{\vect{z} ^{\vect{n}} }{\vect{n}!} \frac{(\sum_k k n_k - 1)!}{(\sum_k kn_k - \sum_k n_k)!} (-1)^{\sum_k n_k}
	\end{equation}
            and the sum is absolutely convergent. 
	\item[(b)] If $\sum_k z_k \exp(a k) >\exp(a)-1$ for all $a>0$, then
                   \begin{equation*}
                       \sum_{\vect{n} \in \mathcal{I}^*}\Bigl| \frac{\vect{z} ^{\vect{n}} }{\vect{n}!}\frac{(\sum_k k n_k - 1)!}{(\sum_k kn_k - \sum_k n_k)!}
 \Bigr| =\infty. 
                   \end{equation*}
	\end{enumerate}  
\end{theorem}
The explicit formula for the expansion coefficients is proven in Section~\ref{sec:trees}, convergence is addressed in Section~\ref{sec:convergence}. 

\begin{remark} 
  There is an interesting activity region intermediate between the domain of convergence and the fluid domain. 
  Write $X$ and $Y$ for discrete intervals $\{x,\ldots,x+k-1\}$, $\{y,\ldots,y+j-1\}$.
  Suppose that 
  \begin{equation*} 
    \sup_{X} \frac{1}{|Y|} \sum_{Y:\, Y \cap X \neq \emptyset} |Y| z_{|Y|} = \sup_{k\geq 1} \frac{1}{k} \sum_{j\geq 1} (k+j-1)j z_j = \sum_{j \geq 1} j^2 z_j <1.
  \end{equation*} 
  Then the infinite volume Gibbs measure can be constructed as the unique stationary measure of a Markov birth and death process, the \emph{loss network}~\cite{fernandez-ferrari-garcia}. The one-dimensional model confirms that the loss network representation works in a domain that is in general larger than the domain of convergence of the activity expansions, since the loss network does not require exponentially decreasing activities $z_k = O(\exp(-a k))$. 
\end{remark} 

\begin{theorem} \label{thm:vir-discrete}
   Let $\vect{z}=(z_k)_{k\in \N}$ be non-negative activities in the fluid domain. Then the partial derivatives $\frac{\partial p}{\partial z_k}(\vect{z})$, $k\in \N$, at $\vect{z}$ exist and we have 
	\begin{equation*} 
		p(\vect{z})  =\log \Bigl(1+ \frac{\sum_k \rho_k (\vect{z})}{1- \sum_k k \rho_k (\vect{z})}\Bigr),  \quad
		z_k  = \frac{\rho_k(\vect{z}) \exp[ (k-1) p(\vect{z})]}{1- \sum_k k \rho_k(\vect{z})}
	\end{equation*} 
	with $ \sum_k k \rho_k(\vect{z}) <1$. 
\end{theorem}
The proof is analogous to the continuous case treated in Section~\ref{sec:virial} and therefore omitted. A heuristic derivation of the pressure-density and density-activity relations is given in Section~\ref{sec:waals}. 

The pressure-density expansion converges again in all of $\mathcal{D}_{\rm fluid}$.  In the special case when there are only monomers, i.e.,  $z_k=0$ for $k\geq 2$, we recover the well-known equations $p= \log (1+z_1)$, $z_1= \rho_1/(1-\rho_1)$. 

%%%%%%%%%%%%%%%%%%%%%%%%%%%%%%%%%%%%%%%%%%%%%%%%%%%%%%%%55
\section{A fixed point equation} \label{sec:explanations}

In this section we provide two different explanations of the fixed point equations for the pressure, a physical explanation in terms of van der Waals mixtures and a probabilistic explanation in terms of renewal processes. 

\subsection{Van der Waals mixtures}\label{sec:waals}
We assume the reader is familiar with notions from statistical mechanics such as the constant pressure ensemble, and employ common approximations such as $\log N! \approx N(\log N-1)$; for notational simplicity we follow physics conventions and pretend that the approximations are identities. 

Consider first the continuous setting. The partition function of the constant pressure ensemble at fixed number $N_k$ of rods of size $k$ can be computed as 
\begin{equation}  \label{eq:isobaric}
	Q(N_1,N_2,\ldots;p) = \binom{M}{N_1,N_2,\ldots} e^{- p \sum_k N_k \ell_k} \Bigl(\int_0^\infty e^{-p r} \dd r\Bigr)^{M-1} 
\end{equation} 
where $M= \sum_k N_k$ is the total number of rods. The multinomial coefficient represents the number of ways to assign rod types to the $M$ rods labelled from left to right. The integral $\int_0^\infty \exp(- pr) \dd r$ comes from integrating over the spacing between two consecutive rods. Eq.~\eqref{eq:isobaric} yields the Gibbs free energy 
\begin{equation} \label{eq:gibbs-energy}  
	G(N_1,N_2,\ldots; p) = \sum_k N_k \Bigl( \log \frac{p N_k}{\sum_j N_j} 
				+ p \ell_k \Bigr).  
\end{equation}
Alternatively, we can compute directly the Helmholtz free energy 
\begin{equation} \label{eq:free-energy} 
	F(N_1,N_2,\ldots;V) =\sum_k N_k \Bigl( \log \frac{N_k}{V - \sum_j  N_j\ell_j} -1 \Bigr). 
\end{equation} 
(see Lemma~\ref{lem:Z}). Recall $G = F + pV$ with $V = - \partial G/\partial p$, $p = - \partial F/\partial V$, which give
\begin{equation} \label{eq:vwaals} 
   p = \frac{\sum_k N_k}{V -\sum_k \ell_k N_k}. 
\end{equation}
This is  a multi-species variant of the van der Waals equation, and generalizes the well-known equation of state of the single-species Tonks gas~\cite{tonks36}.
Using the relation $\log z_k = \partial G/\partial N_k = \partial F/\partial N_k$,  we obtain 
\begin{equation} \label{eq:zkpk}
	z_k =\frac{p N_k}{\sum_j N_j}\,  \exp(  p \ell_k). 
\end{equation} 
Eqs.~\eqref{eq:vwaals} and~\eqref{eq:zkpk} explain the formulas in Theorem~\ref{thm:vir}. The fixed point equation~\eqref{eq:pressure-equation} is obtained from Eq.~\eqref{eq:zkpk} by multiplying both sides with $\exp( - p\ell_k)$ and summing over $k$. 

For the discrete system, the integral in Eq.~\eqref{eq:isobaric} has to be replaced by a geometric sum. The Gibbs energy~\eqref{eq:gibbs-energy} becomes 
\begin{equation*} 
  G(N_1,N_2,\ldots;p)= \sum_k N_k \Bigl( \log \frac{(1- \exp(-p))N_k}{\sum_j N_j} + pk \Bigr).
\end{equation*} 
The equation of state is 
\begin{equation*}  
 p = \log \Bigl( 1 + \frac{\sum_k N_k}{V - \sum_k k N_k}\Bigr),
\end{equation*} 
and the density-activity relation is 
\begin{equation*}
  z_k = \frac{(1-\exp(-p))N_k}{\sum_j N_j}\, \exp(pk),
\end{equation*} 
compare Theorem~\ref{thm:vir-discrete}. We multiply with $\exp(-pk)$, sum over $k$, and obtain the fixed point equation~\eqref{eq:fixed-point-discrete}.

\subsection{Renewal theory}
As done in~\cite{ioffe-velenik-zahradnik06}, one-dimensional polymer partition functions can be treated by renewal theory~\cite[Chapter XI]{fellervol2}. The key idea is to reinterpret the line as a time axis and the starting point of a rod as an event (a light bulb breaks and has to be renewed), and the intervals between two events as waiting or interrarival times. The Gibbs measure is invariant with respect to a suitable rescaling of the activities~\cite{gruber-kunz71}. If the activities can be rescaled in such a way that they define a probability measure on waiting times (the light bulb's lifetime), then the canonical partition function is identified with the probability of a certain event and probabilistic techniques apply. The crucial point now is that this rescaling is possible if and only if the fixed point equation has a solution. 

For details, consider first the discrete case. For $X\subset \Z$ set $\Phi(X)= z_{k}$ if $X= \{x,x+1,\ldots,x+k-1\}$ is a discrete interval of cardinality $k\geq 2$,  $\Phi(X) = 1+z_1$ if $X$ has cardinality $1$, and $\Phi(X)=0$ otherwise. Then 
\begin{equation*} 
  \Xi_L(\vect{z}) = \sum_{\{X_1,\ldots,X_D\}} \Phi(X_1) \cdots \Phi(X_D) 
\end{equation*} 
where the sum runs over all partitions of $\{0,1,\ldots,L-1\}$. An element $X_j= \{x\}$ of cardinality $1$ corresponds to either an unoccupied lattice site or a site occupied by a rod $\{x\}$ of cardinality $1$. Rescaling the activities as $\Phi_\xi(X) = \xi^{\#X} \Phi(X)$ for some $\xi>0$ multiplies the grand-canonical partition function $\Xi_L(\vect{z})$ by $\xi^L$ and leaves the associated probability measure unchanged~\cite{gruber-kunz71}. We wish to choose $\xi$ so that 
\begin{equation}  \label{eq:renewal-discrete}
  (1+z_1) \xi + \sum_{k\geq 2}z_k \xi^k = \sum_{k=1}^\infty \Phi_\xi(\{0,\ldots, k-1\}) =1.
\end{equation}
Substituting $\xi = \exp(−p)$ we obtain the discrete fixed point equation (6). For
activities in the fluid domain, there is therefore a unique solution ξ that stays away
from the radius of convergence, so that the expected  bulb lifetime $\mu  = (1 + z_1)\xi +
\sum_{k\geq 2} k z_k \xi^k$ is finite. Renewal theory then tells us that $\xi^L \Xi_{L}(\vect{z})$ converges to $1/\mu$: the probability that a light bulb has to be renewed at time L given that there was a renewal at time $0$ converges to the
inverse of the average bulb lifetime. In particular, $\frac{1}{L} \log \Xi_L(\vect{z})\to -\log \xi$ and it
follows that the pressure is indeed $p = - \log \xi$. 

For the continuous system we multiply the partition function with $\exp( - \lambda L)$ and distribute the additional factor $\exp(- \lambda L)$ over the rods and empty spaces between them. This attributes the weight $\sum_k z_k \mathbf{1}(r\geq \ell_k)\exp( - \lambda r)$ to an interval going from the starting point of a rod to the next one (instead of treating emtpy space as monomers, we attach an empty interval to the preceding rod). These weights define a probability distribution if 
\begin{equation*}	
	\int_0^\infty \sum_k z_k \mathbf{1}(r\geq \ell_k) \exp( - \lambda r) \dd r= \frac{1}{\lambda} \sum_k  z_k \exp( - \lambda \ell_k) = 1,
\end{equation*} 
and we recognize the fixed point equation~\eqref{eq:pressure-equation}.

%%%%%%%%%%%%%%%%%%%%%%%%%%%%%%%%%%%%%%%%%%%%%%%%%%%%%%%5
\section{Computation of the pressure and packing fraction} \label{sec:fixed-point}

Here we prove Theorems~\ref{thm:pressure} and~\ref{thm:pack-frac} for the continuous system and Theorem~\ref{thm:pressure-discrete} for the discrete system. 

\begin{proof}[Proof of Theorem~\ref{thm:pressure}]
	For $\lambda \geq 0$, let $F(\lambda):= \int_0^\infty \exp( - \lambda  L) \Xi_L(\vect{z}) \dd L$ be the Laplace transform of the grand-canonical partition function with respect to the system length $L$.
	We note that $F(\lambda)< \infty$ for $\lambda>p(\vect{z})$ and $F(\lambda) = \infty$ for $\lambda< p(\vect{z})$. Let $\theta(\vect{z})$ be the solution to $\sum_k z_k \exp( \ell_k \theta) = - \theta$, if the solution exists, and $\theta(\vect{z}) = \theta^*$ otherwise. 

	Ordering rods from left to right and singling out the right end point $x$ of the left-most rod,  we see that 
	the partition function satisfies the renewal equation \cite[Chapter XI]{fellervol2}
	\begin{equation*} 
		\Xi_L(\vect{z}) = 1 + \sum_k z_k \int_{\ell_k}^{L} \Xi_{L-x}(\vect{z}) \dd x.		
	\end{equation*}
	It follows that for all $\lambda \geq 0$, 
	\begin{equation*}
		F(\lambda) = \frac{1}{\lambda} \Bigl(1+ F(\lambda ) \sum_k  z_k e^{-\lambda \ell_k} \Bigr).
	\end{equation*} 
	If $F(\lambda)< \infty$, then $\sum_k  z_k e^{-\lambda \ell_k}$ 
	must be finite as well, thus $\lambda \geq - \theta^*$ and 
	\begin{equation} \label{eq:lambdaf}
		\Bigl( \lambda - \sum_k z_k \exp(- \ell_k \lambda) \Bigr) F(\lambda) = 1. 
	\end{equation} 
	Since $F(\lambda)>0$, we deduce $\sum_k z_k \exp(- \lambda \ell_k) \leq \lambda$ which implies that  $ - \lambda \leq \theta(\vect{z})$. 
	Conversely, if $- \lambda < \theta(\vect{z})$, then $\lambda>\sum_k z_k \exp( - \lambda \ell_k)$
        and $F(\lambda)$ is the unique solution to Eq.~\eqref{eq:lambdaf}, so in particular $F(\lambda)<\infty$. 

	We have shown that $F(\lambda) < \infty$ for $- \lambda <  \theta(\vect{z})$ and $F(\lambda)= \infty$ for $- \lambda> \theta(\vect{z})$. It follows that $p(\vect{z}) = - \theta(\vect{z})$. 
\end{proof} 

\begin{proof}[Proof of Theorem~\ref{thm:pressure-discrete}] 
  For the discrete system, the grand-canonical partition function satisfies the discrete renewal equation~\cite{ioffe-velenik-zahradnik06}
  \begin{equation}\label{eq:recurrence}
    \Xi_L(\vect{z}) = (1+z_1)\Xi_{L-1}(\vect{z}) + \sum_{k=2}^{L-1} z_k \Xi_{L-k}(\vect{z}) + z_L.
  \end{equation} 
  Eq.~\eqref{eq:lambdaf} for the Laplace transform in the continuous setting is replaced by a relation between generating functions 
  \begin{equation*} 
    1 + \sum_{L=1}^\infty \Xi_L(\vect{z}) \xi^L = \Bigl( (1+z_1) \xi + \sum_{k=2}^\infty z_k \xi^k \Bigr) \Bigl( 1+ \sum_{L=1}^\infty \Xi_L(\vect{z}) \xi^L\Bigr)
  \end{equation*} 
  which yields 
  \begin{equation*} 
     1 + \sum_{L=1}^\infty \Xi_L(\vect{z}) \xi^L =\frac{1}{ 1- (1+z_1) \xi - \sum_{k=2}^\infty z_k \xi^k}.
  \end{equation*} 
 The radius of convergence $R$ on the left-hand side is given by the solution $R<\exp( \theta^*)$ to $(1+z_1)R + \sum_{k=2}^\infty z_kR ^k =1$ if the solution exists, or equal to $\exp( \theta^*)$ if it does not. Since $p=- \log R$, we are done.   
\end{proof}

Now we come to the limit law for the packing fraction of the continuous systems. In principle, we could deduce it from Theorem~\ref{thm:pack-frac}, using  $\mu$-dependent activities $z_k(\mu) = z_k \exp( \ell_k\mu)$ and general theorems relating limit laws to differentiability and strict convexity of the pressure as a function of $\mu$. We prefer to follow a more direct approach, which has the advantage of shedding additional light on the value of the pressure. In particular, Lemma~\ref{lem:pvar} below expresses the pressure in terms of a variational problem whose minimizer corresponds to the limiting packing fraction. 

We start by computing the multicanonical partition function. 
For $\vect{N}= (N_k)_{k\in \N} \in \mathcal{I}^*$ and $L>0$, let 
\begin{equation} \label{eq:Z-def}
	Z_L(N_1,N_2,\ldots):=  \frac{1}{\prod_k N_k!} \int_{0}^{L-\ell_1} \dd x_{11} \cdots \int_0^{L-\ell_k} \dd x_{kN_k}\cdots  \mathbf{1}(\text{rods do not overlap})
\end{equation}
so that $\Xi_L(\vect{z}) = 1+ \sum_{\vect{N} \in \mathcal{I}^*} \vect{z}^{\vect{N}} Z_L(\vect{N})$. The next lemma is a variant of the representation of the canonical partition function for non-overlapping rods of fixed length \cite{tonks36}. 

\begin{lemma} \label{lem:Z}
	For every $\vect{N} \in \mathcal{I}^*$ and every $L>0$ such that $\sum_k N_k\ell_k < L$, we have  
	\begin{equation} \label{eq:Z}
		Z_L(N_1,N_2,\ldots) = \frac{(L- \sum_k N_k \ell_k)^{\sum_k N_k}}{ \prod_k (N_k!)}.  
	\end{equation} 
\end{lemma} 

\begin{proof}
	The factorials in Eq.~\eqref{eq:Z-def} can be dropped if we integrate only over sectors where rods of a given type are labelled from left to right, i.e., $x_{k1} \leq \cdots \leq x_{kN_k}$ for all $k$. Write $M:= \sum_k N_k$. The partition function $Z_L(\vect{N})$ is a sum of integrals of the form
	\begin{equation*} 
		\int_{\ell_{k(1)}+\cdots+\ell_{k(M)}}^L \dd x_M \int_{\ell_{k(1)}+\cdots+\ell_{k(M-1)}}^{x_M-\ell_{k(M)}} \dd x_{M-1} \cdots \int_{\ell_{k(1)}}^{x_2- \ell_{k(2)}}  \dd x_1.
	\end{equation*}
	Here $x_j$ represent the end points of rods, and the sum is over color assignments $k(1),\ldots,k(M)$ compatible with $N_1,N_2,\ldots$. There are $\binom{M}{N_1,N_2,\ldots}$ such assignments, and each integral equals $M!^{-1} (L- \sum_k N_k \ell_k)^M$. 
\end{proof} 

Next, we note that Eq.~\eqref{eq:Z} can be reinterpreted in terms of Poisson random variables. 
Given $S>0$ and $\theta \in (-\infty, \theta^*)$, let 
$N_k(\omega)$, $k \in \N$, 
 be independent Poisson random variables $N_k(\omega) \sim \mathrm{Poiss}( (L-S) z_k \exp( \theta \ell_k))$ defined on some common probability space $(\Omega,\mathcal{F}, \P_{\theta,S})$. A straightforward computation shows that if $N_1,N_2,\ldots$ are integers (in $\N_0$) such that $\sum_k N_k \ell_k =S$, then 
\begin{multline} \label{eq:poisson}  
	\boldsymbol{z}^{\boldsymbol{N}}Z_L(N_1,N_2,\ldots)\\
 = \exp\Bigl( (L-S) \sum_k z_k e^{\theta \ell_k} - \theta S \Bigr) 
			\P_{\theta,S}\Bigl( \forall k \in \N:\  N_k(\omega)=N_k \Bigr). 
\end{multline} 
Note that the right-hand side depends on $\theta$, but the left-hand side does not. Eq.~\eqref{eq:poisson} suggests that in the grand-canonical ensemble, the packing fraction satisfies a large deviations principle with convex rate function $I(\sigma)+ p(\boldsymbol{z})$ where 
\begin{equation}
  I(\sigma):= \sup_{\theta<\theta^*}\Bigl[  \theta \sigma -(1-\sigma) g(\theta) \Bigr].
\end{equation}
Remember $g(\theta)$ and $\theta^*$ from Eqs.~\eqref{eq:thetastar} and~\eqref{eq:g}.
Let $u^*:= \lim_{\theta \nearrow\theta^*} g'(\theta)$. Set 
\begin{equation*} 
  \sigma^*:= \begin{cases}
    \frac{u^*}{1+ u^*}, & \quad \text{if}\ u^*<\infty,\\
    1, &  \quad \text{if}\ u^*=\infty.
\end{cases}
\end{equation*}

\begin{lemma}\label{lem:pvar}
  We have 
  \begin{equation*} 
    \min_{\sigma\in [0,1]} I(\sigma) = -  p(\boldsymbol{z}).
  \end{equation*} 
  Moreover 
  \begin{enumerate}
    \item [(a)] In the fluid regime $I(\sigma)$ has a unique minimizer $\sigma(\boldsymbol{z}) \in (0,\sigma^*)$ given by $\sigma(\boldsymbol{z}) = [\sum_k \ell_k z_k \exp( - \ell_k p(\boldsymbol{z}))]/[1 + \sum_k \ell_k z_k \exp( - \ell_k p(\boldsymbol{z}))]$. 
     \item [(b)] In the transition regime the minimizers of $I(\sigma)$ consist precisely of the elements of $[\sigma^*,1]$. 
     \item [(c)]In the close-packing regime $I(\sigma)$ has the unique minimizer $\sigma = 1$. 
  \end{enumerate}
\end{lemma}

\begin{proof}
  	Let $\phi(u):= \sup_{\theta \in \R} [\theta u - g(\theta)]$ be the Legendre transform of $g(\theta)$, where we agree $g(\theta) = \infty$ for $\theta>\theta^*$.  
In $(0,u^*)$, $\varphi(u)$ is strictly convex and smooth; its derivative $\theta = \varphi'(u)$ solves $g'(\theta) = u$, and as  $u\searrow 0$, we have $\phi(u)\to 0$ and $\theta = \phi'(u) \to -\infty$. 
 On $[u^*,\infty)$, $\varphi$ is affine with slope $\theta^*$. 
We note 
\begin{equation*} 
  \inf_{\sigma \in (0,1)} I(\sigma) = \inf_{\sigma \in (0,1)} (1- \sigma) \phi\bigl( \frac{\sigma}{1-\sigma}\bigr) 
   = \inf_{u>0} \frac{\phi(u)}{1+u}. 
\end{equation*} 
For $u\in (0,u^*)$, we have
\begin{equation} \label{eq:deriv}  
  \begin{aligned}
	\frac{\dd}{\dd u} \frac{\phi(u)}{1+u} &= \frac{ \phi'(u)  (1+u) - \phi(u)}{(1+u)^2},\\
   \frac{\dd}{\dd u}\bigl( \phi'(u)  (1+u) - \phi(u)\bigr) 
		&= \phi''(u) (1+u) >0.
   \end{aligned} 
\end{equation} 
Hence the derivative of $\phi(u)/(1+u)$ vanishes at $u$  if and only if $\phi'(u)(1+u) = \phi(u)$. Write $\theta = \phi'(u)$ and remember $g'(u) = \theta$. Then the equation for a vanishing derivative becomes $\theta u + \theta = \phi(u) = \theta u - g(\theta)$, i.e., $\theta = - g(\theta)$ and we recognize our old fixed point equation. 

In the fluid regime the fixed point equation has the unique solution $\theta(\boldsymbol{z}) = - p(\boldsymbol{z})<\theta^*$, and $u = g'(\theta(\boldsymbol{z})) = \sum_k \ell_k z_k \exp( - \ell_k p(\boldsymbol{z})) \in (0,u^*)$ is a strict local minimizer of $\phi(u)/(1+u)$ in $(0,u^*)$. Correspondingly $\sigma = u/(1+u) \in (0,\sigma^*)$ is a strict local minimizer of $I(\cdot)$ in $(0,\sigma^*)$. But $I(\cdot)$, as the supremum of a family of affine functions, is convex, so $\sigma$ is the unique minimizer of $I(\cdot)$ in all of $(0,1)$, and 
\begin{equation*} 
  \inf_{\sigma' \in [0,1]} I(\sigma')= I(\sigma) = \frac{\phi(u)}{1+u}= \phi'(u) = \theta = - p. 
\end{equation*} 
This proves part (a) of the lemma. 

In the transition regime  a variant of the  previous argument shows that $I(\sigma)$ is decreasing on $(0,\sigma^*)$ and $I(\sigma^*)= \theta^*= - p(\boldsymbol{z})$. If $\sigma^*=1$, we are done. If $\sigma^*<1$, then $u^*<\infty$ and for $u>u^*$ 
\begin{equation*} 
  \frac{\phi(u)}{1+u} = \frac{\phi(u^*) + \theta^*(u-u^*)}{1+u} = \frac{\theta^* u^* - g(\theta^*) + \theta^*(u- u^*)}{1+u} = \theta^*= - p(\boldsymbol{z}).
\end{equation*} 
It follows that $I$ is constant on $[\sigma^*,1]$. This proves (b). 

In the close-packing regime the fixed point equation has no solution and $I(\sigma)$ is decreasing on $(0,\sigma^*)$. On $(\sigma^*,1)$ or $(u^*,\infty)$ we have 
\begin{equation*} 
  \frac{\phi(u)}{1+u} = \frac{- g(\theta^*) + \theta^*u}{1+u} = \theta^* + \frac{- g(\theta^*) - \theta^*}{1+u}
\end{equation*} 
which is decreasing in $u$. Thus $I(\sigma)$ is decreasing in all of $(0,1)$ and 
\begin{equation*} 
 I(1) = \inf_{\sigma \in (0,1)} I(\sigma)= \theta^*= - p(\boldsymbol{z}).
\end{equation*} 
This proves part (c). 
\end{proof}

\begin{lemma} \label{lem:ldp}
  Let $[\sigma_1,\sigma_2] \subset [0,1]$. Then 
  \begin{equation*} 
      \limsup_{L\to \infty} \frac{1}{L} \log \sum_{\boldsymbol{N}=(N_1,N_2,\ldots)} \boldsymbol{z}^{\boldsymbol{N}} Z_L(\boldsymbol{N}) \mathbf{1}\Bigl( \tfrac{1}{L}\sum_{k} N_k \ell_k  \in [\sigma_1,\sigma_2] \Bigr)  
\leq - \min_{\sigma \in [\sigma_1,\sigma_2]} I(\sigma). 
    \end{equation*} 
\end{lemma} 

\begin{proof} 
  We first estimate contributions from small intervals $[\sigma, \sigma + \eps]\subset [\sigma_1,\sigma_2]$.  Abbreviate $S= \sum_k N_k \ell_k$ and remember the measures $\P_{\theta, S}$ defined above Eq.~\eqref{eq:poisson}. 
Observe 
\begin{equation}  \label{eq:ldpbound}
    \vect{z}^{\vect{N}} Z_L(\boldsymbol{N}) \leq \exp \Bigl( - \theta S + (L-L\sigma) g(\theta) \Bigr) \P_{\theta,L\sigma} \Bigl(\forall k:\,N_k(\omega) = N_k\Bigr)
 \end{equation} 
for all $\theta$. We choose $\theta$ as the maximizer $\theta(\sigma)$ of $\theta \sigma - (1- \sigma) g(\theta)$. If $\theta \geq 0$, we have $- \theta S \leq - \theta \sigma L$ for all $S\in [\sigma L, (\sigma + \eps)L]$ and 
\begin{multline*} 
  \sum_{\boldsymbol{N}} \boldsymbol{z}^{\boldsymbol{N}} Z_L(\boldsymbol{N}) \mathbf{1}\bigl( S/L  \in[\sigma, \sigma + \eps]\bigr) \\
 \leq \exp\Bigl( - L I(\sigma) \Bigr) \P_{\theta, L\sigma}\Bigl( \tfrac{1}{L} \sum_k N_k(\omega) \ell_k \in [\sigma, \sigma + \eps] \Bigr)
\end{multline*} 
hence 
\begin{equation*}
  \limsup_{L\to \infty} \frac{1}{L} \log\sum_{\boldsymbol{N}} \boldsymbol{z}^{\boldsymbol{N}} Z_L(\boldsymbol{N})\mathbf{1}\bigl( S/L\in [\sigma, \sigma+\eps] \bigr) \leq - I(\sigma).
\end{equation*} 
If $\theta \leq 0$, we have instead $-\theta S \leq - \theta L( \sigma + \eps)$ and 
\begin{equation*}
  \limsup_{L\to \infty} \frac{1}{L} \log\sum_{\boldsymbol{N}} \boldsymbol{z}^{\boldsymbol{N}} Z_L(\boldsymbol{N})\mathbf{1}\bigl( S/L\in [\sigma, \sigma+\eps] \bigr) \leq - I(\sigma) - \theta \eps.
\end{equation*} 
If $\sigma_1>0$, then $\theta(\sigma) \geq \theta(\sigma_1)=:\theta_1$ stays bounded away from $-\infty$. We split $[\sigma_1,\sigma_2]$ into $m$ slices of width $\eps = (\sigma_2 - \sigma_1)/m$ and obtain the upper bound
\begin{equation*} 
  - \min_{\sigma \in [\sigma_1,\sigma_2]} I(\sigma) + \eps \max(- \theta_1,0).
\end{equation*} 
Letting $\eps \to 0$ then yields the desired bound.

If $\sigma_1 =0$ we have to proceed more carefully. Fix 
$\sigma_0$ such that $\theta_0 = \theta(\sigma_0) <\min(\theta^*,0)$. For $S\in [\sigma, \sigma + \eps] \subset [0,\sigma_0]$ we apply the bound~\eqref{eq:ldpbound} to $\theta = \theta(\sigma+\eps)\leq \theta_0 <0$ and find 
\begin{equation*} 
   \limsup_{L\to \infty} \frac{1}{L} \log\sum_{\boldsymbol{N}} \boldsymbol{z}^{\boldsymbol{N}} Z_L(\boldsymbol{N})\mathbf{1}\bigl( S/L\in [\sigma, \sigma+\eps] \bigr) \leq - I(\sigma+\eps) + \eps g(\theta).
\end{equation*} 
Note that $g(\theta) \leq g(\theta_0)$ is bounded uniformly in $\sigma\in [0,\sigma_0]$. Splitting $[0,\sigma_0]$ into slices of width $\eps$ and letting first $L\to \infty$ and then $\eps \to 0$, we obtain an upper bound for the sum over packing fractions in $[0,\sigma_0]$ of the desired form. This works for every sufficiently small $\sigma_0$. Combined with the bounds for intervals $[\sigma_1,\sigma_2]$ with $\sigma_1>0$, this proves the lemma.   
\end{proof}

Now we can prove Theorem~\ref{thm:pack-frac}. 

\begin{proof}[Proof of Theorem~\ref{thm:pack-frac}]
  By Lemma~\ref{lem:pvar}, in both situations (a) and (b) of the theorem, $I(\sigma)$ has a unique minimizer $\sigma_0$, given by $\sigma(\boldsymbol{z})$ or $\sigma =1$. Let $\eps>0$. Then by Lemma~\ref{lem:ldp} the grand-canonical probability of seeing a packing fraction in $[0, \sigma_0-\eps]$ or $[\sigma_0 + \eps, 1]$ is bounded by 
\begin{equation*} 
  \exp\Bigl(- L \min_{|\sigma- \sigma_0|>\eps} \bigl( I(\sigma) + p(\boldsymbol{z})\bigr) +o(L)\Bigr)
\end{equation*} 
which goes to zero exponentially fast as $L\to \infty$. 
\end{proof}

%%%%%%%%%%%%%%%%%%%%%%%%%%%%%%%%%%%%%%%%%%%%%%%%%%%%%%%%%%%%%%%%%%%%%%%%%%%%%%%%%%%%5
%%%%%%%%%%%%%%%%%%%%%%%%%%%%%%%%%%%

\section{Computation of the pressure-activity expansion} \label{sec:trees}

In this section we show that the fixed point equation for the pressure determines the coefficients of the (formal) activity expansion uniquely. The convergence of the expansion is dealt  with in Section~\ref{sec:convergence}. 
Because of alternating signs, it is convenient to work with $F(\vect{z})= - p (- \vect{z})$.

\subsection{Continuous system} The functional equation~\eqref{eq:pressure-equation} for the pressure translates into 
\begin{equation} \label{eq:F} 
	F(\vect{z}) = \sum_k z_k \exp( \ell_k F(\vect{z})). 
\end{equation}
Let $F(\vect{z}) = \sum_{\boldsymbol{n}} a(\boldsymbol{n}) \vect{z}^{\vect{n}}$ be a formal power series. The right-hand side of Eq.~\eqref{eq:F} can be read as a formal power series whose coefficients are polynomials of $a(\boldsymbol{n})$; we say that Eq.~\eqref{eq:F} holds in the sense of formal power series if the coefficients are equal, i.e., 
\begin{equation*} 
  \forall \boldsymbol{n}\in \mathcal{I}:\   [ \vect{z}^{\vect{n}}]F(\vect{z}) = [ \vect{z}^{\vect{n}}] \sum_k z_k \exp( \ell_k F(\vect{z})).
\end{equation*} 
The notation ``$[ \vect{z}^{\vect{n}}]F(\vect{z})$'' stands for ``coefficient of the monomial $\prod_k z_k^{n_k}$ in the formal power series $F(\vect{z})$''.

\begin{prop} \label{prop:tonks-formula}
   The formal power series $F(\boldsymbol{z}) = \sum_{\boldsymbol{n}} a(\boldsymbol{n}){\boldsymbol{z}}^{\boldsymbol{n}}$ satisfies Eq.~\eqref{eq:F} if and only if $a(\boldsymbol{n}) = \frac{1}{\boldsymbol{n}!} (\sum_k \ell_k n_k)^{\sum_k n_k -1}$ for all $\boldsymbol{n}\in \mathcal{I}^*$, and $a(\boldsymbol{0})=0$. 
\end{prop}

We give two independent proofs, a direct proof based on the fixed point equation and Lagrange-Good inversion~\cite{good60}, and a combinatorial proof based on a connection between Eq.~\eqref{eq:F} and recurrence relations between weighted trees. The first proof is more robust and is easily adapted to the discrete model.  The second proof is of interest as it provides yet another example for the deep connections between combinatorics and cluster expansions \cite{faris10}. 

\begin{proof}[Proof of Proposition~\ref{prop:tonks-formula} via Lagrange-Good inversion]
  The identity~\eqref{eq:F} is equivalent to a system of equations for the expansion coefficients $a(\vect{n})$. The system is nonlinear but has a triangular form, i.e., $a(\vect{n})$ is expressed in terms of coefficients $a(\vect{m})$ corresponding to total degree $\sum_k m_k <\sum_k n_k$. The existence and uniqueness of a solution is therefore easily proven by induction over $\sum_k n_k$. Thus we are left with the computation of the coefficients. Clearly $a(\vect{0}) = F(\vect{0}) = 0$. For non-zero multi-indices, write $w_k(\vect{z}):= z_k \exp( \ell_k F(\vect{z}) )$. We have 
	\begin{equation*} 
		z_k = w_k \exp( - \ell_k \sum_j w_j), \quad k\in \N. 
	\end{equation*} 
	We look for the expansion of $F(\vect{z} ) =\sum_j w_j(\vect{z}) $ in powers of the $z_k$s. A formal computation suggests 
	\begin{align*} 
		[\vect{z}^{\vect{n}}] (\sum_j w_j(\vect{z})) & = 
			\prod_k \Bigl(\frac{1}{2\pi \mathrm{i}} 
			\oint \frac{\dd z_k}{z_k^{n_k+1}} \Bigr) (\sum_j w_j(\vect{z})) \\
			& =  \prod_k \Bigl(\frac{1}{2\pi \mathrm{i}} 
			\oint \frac{\dd w_k}{z_k(\vect{w})^{n_k+1}} \Bigr) (\sum_j w_j) \det \Bigl( \bigl(\frac{\partial z_k}{\partial w_j}\bigr)_{k,j} \Bigr) \\
			& = \prod_k \Bigl(\frac{1}{2\pi \mathrm{i}} 
			\oint \frac{\dd w_k}{w_k^{n_k+1}} \Bigr) (\sum_j w_j)
				e^{\sum_k (n_k+1) \ell_k \sum_j w_j} \det \Bigl( \bigl(\frac{\partial z_k}{\partial w_j}\bigr)_{k,j} \Bigr)
	\end{align*} 
	The Lagrange-Good inversion formula \cite{good60,ehrenborg-mendez} says that indeed 
	\begin{equation*} 
		[\vect{z}^{\vect{n}}] (\sum_{j=1}^\infty w_j(\vect{z})) 
		 = [\vect{w}^{\vect{n}}] \, \Bigl\{ (\sum_{j=1}^\infty w_j)
				e^{\sum_{k\in {\rm supp}\, \vect{n}} (n_k+1) \ell_k \sum_{j=1}^\infty w_j} \det \Bigl( \bigl(\frac{\partial z_k}{\partial w_j}\bigr)_{k,j \in \mathrm{supp}\, \vect{n}} \Bigr) \Bigr\}. 			
	\end{equation*} 
	The partial derivative is
	$	\frac{\partial z_k }{\partial w_\ell} = \bigl( \delta_{k \ell} - \ell_k w_k\bigr) \exp\bigl(- \ell_k \sum_j w_j \bigr)$
	and the determinant equals
	\begin{equation*} 
		\det\Bigl(\delta_{k\ell} - \ell_k w_k\Bigr)_{k,\ell \in {\rm supp}\, \vect{n}} = 1 - \sum_{k\in {\rm supp}\, \vect{n}} \ell_k w_k.
	\end{equation*} 
	Note for matrices $A$ of rank one, $\det(\mathrm{id} + A) = 1+ \sum_{j=1}^n a_{jj}$.
	 It follows that 
	\begin{align*}
		 [\vect{z}^{\vect{n} } ] (\sum_j w_j) &= [\vect{w}^{\vect{n}}]   (\sum_j w_j) (1- \sum_j \ell_j w_j) e^{(\sum_k n_k \ell_k) (\sum_j w_j)}. 
	\end{align*}  
	Now 
	\begin{align*}
		[\vect{w}^{\vect{m}}] (\sum_j w_j) e^{(\sum_k n_k \ell_k) (\sum_j w_j)} 
		& =  [\vect{w}^{\vect{m}}] \, \frac{1}{(\sum_j m_j-1)!} \bigl( \sum_k n_k \ell_k)^{\sum_j m_j-1 } (\sum_j w_j)^{\sum_j m_j} \\  
		& = \frac{\sum_j m_j}{\prod_j (m_j!)} \bigl( \sum_k n_k \ell_k)^{\sum_j m_j-1 }.
	\end{align*} 
	Therefore ($|\vect{n}|=\sum_k n_k$)
	\begin{align*} 
		[\vect{z}^{\vect{n}}] (\sum_j w_j) & = \frac{|\vect{n}|} {\vect{n} !} 
			\bigl( \sum_k n_k \ell_k\bigr)^{|\vect{n}|-1 } 
			- \sum_j \ell_j \frac{|\vect{n}|-1}{\vect{n}! } n_j \bigl(\sum_k n_k \ell_k\bigr)^{|\vect{n}|-2} \\ 
		& = \frac{1}{\vect{n}!} \bigl( \sum_k n_k \ell_k\bigr)^{|\vect{n}|-1}. \qedhere
	\end{align*} 	 
\end{proof} 

For the combinatorial proof of Proposition~\ref{prop:tonks-formula} we identify $F(\vect{z})$ as an exponential generating function of colored labelled weighted trees. 
First we need some definitions. A \emph{colored set} is a pair $(V,c)$ consisting of a finite set $V$ and a color map $c:V \to \N$.  We consider two colored sets as equivalent if for each color $k$, they have the same number $n_k$ of elements of a given color; equivalently, if there is a color-preserving bijection between them. Given $\vect{n} \in \mathcal{I}^*$, let $(V,c)$ be one representative of the equivalence class with color numbers $n_k$. For example, we can choose $V = \{(1,1),\ldots,(1,n_1), (2,1),\ldots,(2,n_2),\ldots\}$ with the natural color mapping $c(k,n) = k$.  Let $\mathcal{T}(\vect{n})$ be the collection of rooted trees on $V$.  We think of such trees as  \emph{colored rooted trees labelled within colors}. To each such tree we assign a weight given by 
\begin{equation*}
	w(\gamma):= \prod_{e\in E(\gamma)} \ell_{c(\text{father in }e)}.
\end{equation*} 
Put differently, the weight of an edge originating in a father with color $k$ is $\ell_k$, and the weight of the tree is the product of the edge weights. The corresponding exponential generating function is 
\begin{equation*} 
	G(\vect{z}) = \sum_{\vect{n}\in \mathcal{I}^*} \frac{ \vect{z}^{\vect{n}}}{\vect{n}!} \sum_{\gamma \in \mathcal{T}(\vect{n})} w(\gamma). 
\end{equation*}

\begin{proof}[Proof of Proposition~\ref{prop:tonks-formula} via combinatorics] 
The combinatorial proof is in two steps: 1. show $F(\vect{z}) = G(\vect{z})$, 2. compute the sum of weights $\sum_{\gamma \in \mathcal{T}(\vect{n})} w(\gamma)$.

First we note that $G(\vect{z})$ satisfies the same functional equation as $F$. This can be seen by summing over the color of the root, and then over the subtrees attached to the root, see~\cite{good61} for similar relations for colored trees.  The functional equation determines the power series coefficients $[\vect{z}^{\vect{n}}] G(\vect{z})$ uniquely, as can be seen for example by an induction over $\sum_k n_k$. As a consequence, $F(\vect{z})= G(\vect{z})$. 

For the computation of the weights, let $d_i$ be the degree of a vertex in a given graph $\gamma$. Each vertex is the father of $d_{i} -1$ children, except the root which has $d_i$ children. Thus we can compute the sum of weights by summing first over all possible degree distributions and then over the choice of the root. We use the formula for the number of labelled (unrooted) trees with a given degree distribution \cite{moon70}, write  $m=|\vect{n}|=\sum_k n_k$, and obtain 
\begin{align*} 
	\sum_{\gamma \in \mathcal{T}(\vect{n})} w(\gamma) 
	& = \sum_{ (d_i)_{1\leq i \leq m}:\, \sum_i d_i = 2 m-2 } \binom{m-2}{d_1-1,\ldots,d_m-1} 
		\Bigl(\prod_{i=1}^m \ell_{c(i)}^{d_i -1} \Bigr) \Bigl( \sum_{i=1}^m \ell_{c(i)} \Bigr) \\ 
	& = \Bigl( \sum_{i=1}^m \ell_{c(i)} \Bigr)^{m-1} = \bigl( \sum_k n_k \ell_k\bigr)^{\sum_k n_k -1}.
\end{align*} 
Here we have implicitly chosen our trees as trees on $V= \{1,\ldots,m\}$; recall that $c:V\to \N$ is a fixed map such that $\# \{i\in V\mid c(i) =k\} =n_k$.
\end{proof} 

%%%%%%%%%%%%%%%%%%%%%%%%%%%%%%%%%%%%%%%%%%%%%%5
\subsection{Discrete system} 
For the discrete system, the fixed point equation after sign flip becomes 
\begin{equation}  \label{eq:F-discrete}
  \exp\bigl( F(\vect{z}) \bigr) = 1 + \sum_{k=1}^\infty z_k \exp\bigl( k F(\vect{z}) \bigr). 
\end{equation} 

\begin{prop} \label{prop:tonks-formula-discrete}
   The formal power series $F(\boldsymbol{z}) = \sum_{\boldsymbol{n}} a(\boldsymbol{n}){\boldsymbol{z}}^{\boldsymbol{n}}$ satisfies Eq.~\eqref{eq:F-discrete} if and only if
\begin{equation*} 
 a(\boldsymbol{n}) = \frac{1}{\prod_k n_k!}\times \frac{ ( \sum_k k n_k -1)!}{(\sum_k kn_k -\sum_k n_k)!}.
\end{equation*} 
 for all $\boldsymbol{n}\in \mathcal{I}^*$, and $a(\boldsymbol{0})=0$. 
\end{prop}

We provide a proof based on Lagrange-Good inversion and leave open the combinatorial interpretation. Note that $G(\vect{z})= \exp( F(\vect{z}))$ can be interpreted as an ordinary generating function for non-labelled, colored planar trees. For example, when all activities except for a given rod length $k$ vanish, Eq.~\eqref{eq:F-discrete} becomes $G= 1+ z_k G^k$ which is the generating function for $k$-ary trees (every vertex has exactly $k$ children, only non-leaf vertices are counted), and the expansion coefficients of $G$ are generalized Catalan numbers~\cite{hilton-pedersen}. When $k=2$, $G=1+z_2G^2$ is the generating function of the standard Catalan numbers $\frac{1}{n+1}\binom{2n}{n}$ and is associated with binary trees. The interpretation of $F = \log G$, however, is less clear. 

\begin{proof} 
  The identity $F = \log (1 + \sum_k z_k \exp(kF))$ leads to a system of equations for the expansion coefficients $a(\vect{n})$ that is triangular, and the existence and uniqueness of the solution can be proven by induction just as for Proposition~\ref{prop:tonks-formula}. Thus we are left with the computation of the coefficients.

  Let $w_k= z_k \exp( k F(\vect{z}))$. Then $\exp(F) = 1+ \sum_{k=1}^\infty w_k$ and $z_k = w_k/(1+\sum_{j=1}^\infty w_j)^k$. We compute 
  \begin{equation*} 
    \frac{\partial z_k}{\partial w_\ell} = \frac{\delta_{k\ell}}{(1+\sum_{j=1}^\infty w_j)^k} - \frac{k w_k}{(1+\sum_{j=1}^\infty w_j)^{k+1}}
  \end{equation*} 
  and for every finite non-empty set $I \subset \N$ 
  \begin{equation*} 
    \det \bigl( \frac{\partial z_k}{\partial w_\ell} \bigr)_{k,\ell \in I}  
      = \frac{1}{(1+ \sum_{j=1}^\infty w_j)^{\sum_{k\in I} k}} \Bigl( 1- \frac{\sum_{k \in I}k w_k}{1+\sum_{j=1}^\infty w_j} \Bigr).
  \end{equation*}
  The identity $a(\vect{0}) = 0$ is obvious. Let $\vect{n} \in \mathcal{I}^*$ be a non-zero multi-index. By Lagrange-Good inversion, we have
 \begin{align*} 
&	[\vect{z}^{\vect{n}}] F(\vect{z}) \\
	& \ 	 = [\vect{w}^{\vect{n}}] \, \Bigl\{ \log(1+ \sum_{j=1}^\infty w_j)
				(1+\sum_{j=1}^\infty w_j)^{\sum_{k\in {\rm supp}\, \vect{n}} k(n_k+1)} \det \Bigl( \bigl(\frac{\partial z_k}{\partial w_j}\bigr)_{k,j \in \mathrm{supp}\, \vect{n}} \Bigr) \Bigr\} \\			
   & \ = [\vect{w}^{\vect{n}}]  (1+\sum_j w_j)^{\sum_k k n_k} \log(1+ \sum_j w_j)   \\
   & \qquad \qquad - \sum_{k \in \supp{\vect{n}}}[\vect{w}^{\vect{n}- \vect{e}_k}] k (1+\sum_j w_j)^{\sum_k k n_k-1 } \log(1+ \sum_j w_j).
\end{align*} 
Let $H_0=0$ and $H_m=1+\frac{1}{2} + \cdots + \frac{1}{m}$ be the harmonic numbers. For $m\in \N_0$ set $f_m(u)= \frac{1}{m!} (1+u)^m ( \log (1+u) - H_m)$. Then $f_{m+1}'(u) = f_m(u)$ and an induction over $m$ leads to 
\begin{equation*} 
  (1+u)^m \log(1+u) = \sum_{k=0}^m \binom{m}{k} (H_m- H_{m-k}) u^k + O(u^{m+1}). 
\end{equation*} 
Let $N= \sum_k k n_k$ and $M= \sum_k n_k$. 
It follows that 
\begin{align*} 
   [\vect{w}^{\vect{n}}] (1+\sum_j w_j)^N \log (1+ \sum_j w_j)  
   &  = \binom{M}{n_1,n_2,\ldots} \times [u^M] (1+u)^N \log(1+ u)\\
    &  =  \binom{M}{n_1,n_2,\ldots} \binom{N}{M} (H_N- H_{N-M}) \\
   & = \frac{1}{\vect{n}!} \frac{N!}{(N-M)!} (H_N- H_{N-M}).
\end{align*} 
Similarly, 
\begin{align*} 
   & [\vect{w}^{\vect{n- \vect{e}_k}}] (1+\sum_j w_j)^{N-1} \log (1+ \sum_j w_j)  \\
   &\qquad \qquad   = \binom{M-1}{n_1,\ldots, n_k-1,\ldots} \times [u^{M-1}] (1+u)^{N-1} \log(1+ u)\\
    & \qquad \qquad  =  \binom{M-1}{n_1,\ldots,n_k-1,\ldots} \binom{N-1}{M-1} (H_{N-1}- H_{N-M}) \\
   &\qquad \qquad  = \frac{n_k}{\vect{n}!} \frac{(N-1)!}{(N-M)!} (H_{N-1}- H_{N-M}).
\end{align*} 
Therefore 
\begin{align*} 
  [\vect{z}^{\vect{n}}] F(\vect{z}) &= \frac{1}{\vect{n}!} \frac{N!}{(N-M)!} (H_N- H_{N-M}) 
  - \sum_k k \frac{n_k}{\vect{n}!} \frac{(N-1)!}{(N-M)!} (H_{N-1}- H_{N-M}) \\
  & = \frac{N!}{\vect{n}! (N-M)!} (H_N- H_{N-1})  = \frac{(N-1)!}{\vect{n}!(N-M)!}.
\end{align*} 

\end{proof}

%%%%%%%%%%%%%%%%%%%%%%%%%%%%%%%%%%%%%%%%%%%%%%%%%%%%%5555
%%%%%%%%%%%%%%%%%%%%%%%%%%%%%%%%%%%%%%%%%%%%%%%%%%%%%%%%%%55

\section{Convergence of the activity expansion} \label{sec:convergence}

\subsection{Continuous system} 
Here we prove Theorem~\ref{thm:cluster-convergence}. The fixed point equation for $p$ and Proposition~\ref{prop:tonks-formula} show that if $p(\boldsymbol{z})$ has a convergent expansion, then that expansion is necessarily given by Eq.~\eqref{eq:tonks}. 
Thus it remains to investigate the convergence of the series. Because of the alternating signs, it is enough to look at
$$F(\vect{z}) = \sum_{\vect{n}\in \mathcal{I}^*} \frac{\vect{z}^{\vect{n}}}{\vect{n}!} 
	\bigl( \sum_k \ell_k n_k\bigr)^{\sum_k n_k -1}$$ 
for non-negative variables $z_k \geq 0$. 

The proof of the divergence criterion is based solely on the fixed point equation.

\begin{proof}[Proof of Theorem~\ref{thm:cluster-convergence}(b)]
	Take variables $z_k \geq 0$. Recall  from Section~\ref{sec:trees} that $F(\vect{z}) = \sum_k z_k \exp( \ell_k F(\vect{z}))$ in the sense of formal power series. If $F(\vect{z})$ is finite, the identity holds as an identity between positive numbers. Since $F\geq 0$, it follows that $\inf_a (\sum_k z_k \exp( \ell_k a) - a) \leq \sum_k z_k \exp( \ell_k F) - F = 0$. Taking the converse, we see that if $\inf_a (\sum_k z_k \exp(\ell_k a) - a) >0$, then $F(\vect{z}) = \infty$. 
\end{proof} 

For the proof of the convergence criterion, we use the explicit expression of the coefficients and Poisson variables (compare the proof of Theorem~\ref{thm:pack-frac} in Section~\ref{sec:fixed-point}). 

\begin{proof}[Proof of the convergence in  Theorem~\ref{thm:cluster-convergence}(a)]
	Take $z_k \geq 0$.
For $V\in \mathcal{V}:= \{\sum_k n_k \ell_k \mid \vect{n} \in \mathcal{I}^*\}$, let 
\begin{equation*}
	h_V:= \sum_{\vect{n} \in \mathcal{I}^*} \frac{\vect{z}^{\vect{n}}}{\vect{n!}} V^{\sum_k n_k} \, 
		 \mathbf{1}\Bigl(\sum_k n_k \ell_k = V\Bigr).
\end{equation*} 
Note $F(\vect{z}) = \sum_{V\in \mathcal{V}} h_V/V.$
  
Assume first that $\sum_k z_k\exp(a \ell_k)<a$ for some $a>0$.  For $V>0$, let 
	$R_k \sim \mathrm{Poiss}(V z_k \exp( a \ell_k))$, $k\in \N$ be independent Poisson random variables defined on some common probability space $(\Omega,\mathcal{F}, \P_V)$.
	Note that 
	\begin{equation*}\label{eq:hv-proba}	
		 h_V =  e^{- \alpha V}\P_V( \sum_k \ell_k R_k(\omega) = V),\quad \alpha:= a-\sum_k z_k\exp(a \ell_k)>0. 
	\end{equation*}
	More generally if $V \in \mathcal{V} \cap [n,n+1)$ for some $n\in \N$, then 
	\begin{align*} 
	  h_V  &\leq \sum_{\vect{n} \in \mathcal{I}^*} \frac{\vect{z}^{\vect{n}}}{\vect{n!}} (n+1)^{\sum_k n_k} \, \mathbf{1}\Bigl(\sum_k n_k \ell_k = V\Bigr)  \\
	      & = \P_{n+1} \Bigl( \sum_k \ell_k R_k(\omega) = V\Bigr) e^{- a V  + (n+1) \sum_k \exp( a \ell_k)} \\	
	      & \leq \P_{n+1} \Bigl( \sum_k \ell_k R_k(\omega) = V\Bigr) e^{a - \alpha(n+1) }. 
	\end{align*}
	Therefore 
	\begin{equation*}
		\sum_{n=1}^\infty\ \sum_{V\in \mathcal{V} \cap [n,n+1)} \frac{h_V}{V} 	
			\leq \sum_{n=1}^\infty \frac{1}{n} e^{a - \alpha (n +1) } < \infty.  				\end{equation*} 	
	Note that if $\ell_k \to 0$, the set $V \in \mathcal{V}\cap [n,n+1)$ can be countably infinite. 
	In addition, 
	\begin{equation*} 
		\sum_{V\in \mathcal{V}:\ V\leq 1} \frac{h_V}{V} \leq 
			\sum_{\vect{n} \in \mathcal{I}^*} \frac{\vect{z}^{\vect{n}}}{\vect{n}!} 
				= \exp \bigl( \sum_k z_k\bigr) -1 < e^a -1< \infty.  
	\end{equation*}
	 It follows that $F(\vect{z}) < \infty$. Next we show, still under the assumption $\sum_k z_k \exp( a \ell_k) <a$, that 
          \begin{equation}\label{eq:Fbounds}
             F(\boldsymbol{z}) < a. %,\quad \sum_k \ell_k z_k e^{\ell_k F(\boldsymbol{z})} <  1. 
  \end{equation} 
         To this aim consider the curve $y= \sum_{k} z_k \exp( \ell_k \theta)$, $\theta \geq 0$. It is the graph of a convex increasing function that starts at $g(0) >0$, i.e., above the line $y = \theta$ and is below the line at $\theta=a$, $g(a)<a$. 
         It must cross the line at some $\theta_1<a$ and possibly crosses it again at some $\theta_2 >a$. Moreover $g'(\theta_1) <1 <g'(\theta_2)$. We already know that $F(\boldsymbol{z})$ solves $g(F) = F$. If the fixed point equation has a unique solution, then $F= \theta_1$ and we are done. If the equation has two solutions, we invoke a continuity argument: for $t\in [0,1]$, write $t \boldsymbol{z}= (tz_k)_{k\in \N}$. Then $F(t\boldsymbol{z})$ solves $\sum_k t z_k \exp( \ell_k \theta) = \theta$. For $t$ close to $1$, the latter equation has two solutions $\theta_1(t)$, $\theta_2(t)$ with $\theta_{1,2}(1) = \theta_{1,2}$ and 
\begin{equation} 
   \sum_k  \ell_k t z_k e^{\theta_1(t) \ell_k} <1< \sum_k \ell_k t z_k e^{\theta_2(t) \ell_k}.
  \end{equation}
It follows that as $t\searrow 0$, either $\theta_2(t)$ ceases to exist or $\theta_2(t) \to \infty$. On the other hand the power series expansion for $F$ converges for all $t\leq 1$ and has no zero-order term, thus $t\mapsto F(t \boldsymbol{z})$ is continuous on $[0,1]$ and converges to $0$ as $t\searrow 0$. It follows that $F(t\boldsymbol{z}) = \theta_1(t)$ for all $t\in [0,1]$. Therefore $F(\boldsymbol{z}) = \theta_1<a$.
% and $\sum_k z_k \exp(\ell_k F(\boldsymbol{z}) = g'(\theta_1) <1$. 
This proves Eq.~\eqref{eq:Fbounds}.

    Now consider the case $\sum_k z_k \exp( a \ell_k) =a$ for some $a>0$. Let $t\in [0,1)$. Then $\sum_k (tz_k) \exp( a\ell_k) = ta <a$, hence 
    \begin{equation*} 
       F(t \vect{z}) = \sum_{\boldsymbol{n}} \frac{\vect{z}^{\vect{n}}}{\vect{n}!} (\sum_k n_k \ell_k)^{\sum_k n_k - 1} t^{\sum_k n_k} <a.
    \end{equation*} 
    We let $t\nearrow 1$, use Abel's theorem and conclude that $F(\vect{z}) \leq a <\infty$. 
%Note that by continuity,  $\sum_k \ell_k z_k \exp(\ell_k F(\vect{z})) \leq 1$. 
\end{proof} 

%%%%%%%%%%%%%%%%%%%%55
\subsection{Discrete system} Here we prove Theorem~\ref{thm:convergence-discrete}. 
Because of the alternating signs and Proposition~\ref{prop:tonks-formula-discrete}, it is enough to look at 
\begin{equation*} 
  F(\vect{z}) = \sum_{\vect{n} \in \mathcal{I}^*} \frac{\vect{z}^{\vect{n}}}{\vect{n}!} \frac{(\sum_k k n_k -1)!}{(\sum_k k n_k - \sum_k n_k)!} 
\end{equation*} 
for non-negative $z_k$. 

The proof of the divergence criterion is based on the fixed point equation and is completely analogous to the continuous setting. 
\begin{proof}[Proof of Theorem~\ref{thm:convergence-discrete}(b)]
  If $F(\vect{z})<\infty$, then it satisfies the fixed point equation $\exp(F(\vect{z}) =1 + \sum_k z_k \exp( k F(\vect{z}))$. Setting $a= F$ we find $\sum_k z_k \exp( a k) \leq  \exp(a) - 1$. It follows that if $\sum_k z_k \exp( ak) >\exp(a)-1$ for all $a$, then $F(\vect{z}) = \infty$. 
\end{proof}

For the sufficiency of the convergence criterion, we have to replace the Poisson distribution employed in the continuous setting by something else. 

\begin{proof}[Proof of the convergence in Theorem~\ref{thm:convergence-discrete}(a)]
  Assume first that 
  \begin{equation*} 
    \frac{1+ \sum_{k\geq 1} z_k \exp(ka)}{\exp(a)} =: \exp( - \alpha) <1
  \end{equation*}
  for some $a>0$. 
  Let 
\begin{equation*} 
    p_0:= \frac{1}{1+\sum_{j\geq 1} z_j \exp(aj)},\quad p_k:=\frac{z_k\exp(ak)}{1+\sum_{j\geq 1} \exp(aj)} \ (k\in \N). 
  \end{equation*} 
   Fix $S\in \N$. Consider $S$ independent random variables with values in $\N_0$ and probability distribution $(p_k)$. We may think of $S$ boxes that are either blank or have the color $k$, and each box chooses its color independently according to the probability weights $p_k$. Let $(n_k)_{k\in \N}$ be a sequence of non-negative integers with $M=\sum_k n_k \leq S$. The probability that  $S-M$ of the boxes are blank and the remaining $M$ boxes have exactly $n_k$ boxes of color $k$, for all $k\in \N$, is given by 
\begin{align*} 
&  \binom{S}{S-M,n_1,n_2,\ldots} p_0^{S-M} \prod_{k\geq 1} p_k^{n_k} \\
  &\quad  =  \binom{S}{S-M,n_1,n_2,\ldots} \frac{1}{(1+\sum_j z_j \exp( aj))^S} \prod_{k\geq 1}(z_k e^{ak})^{n_k}   \\
  & \quad = \frac{S!}{(S-M)!} \frac{\vect{z}^{\vect{n}}}{\vect{n}!} \Bigl( \frac{\exp(a)}{1+ \sum_j z_j \exp(aj)}\Bigr)^S. 
\end{align*} 
It follows that 
\begin{equation*} 
  \sum_{\vect{n} \in \mathcal{I}^*} \frac{\vect{z}^{\vect{n}}}{\vect{n}!} \frac{(\sum_k k n_k -1)!}{(\sum_k k n_k - \sum_k n_k)!} \mathbf{1}\Bigl( \sum_k k N_k =S\Bigr) \leq \frac{\exp(-\alpha S)}{S}
\end{equation*} 
and
\begin{equation*} 
   F(\vect{z}) \leq \sum_{S\geq 1} \frac{\exp( - \alpha S)}{S} <\infty.
\end{equation*} 
In order to cover the case that $\sum_{k \geq 1} z_k \exp( ka) = \exp(a)-1$ for some $a>0$ but ``$<$'' fails, one exploits Abel's theorem and the fixed point equation. The procedure is in two steps: 1. 
 Backtrack to the case that $\sum_k z_k \exp(ka) <\exp(a)-1$. Then $F(\vect{z})$ is finite and  satisfies the fixed point equation $ F = \log(1+ \sum_{k\geq 1} z_k \exp(kF))=:h(F) $. If $h(\theta) =\theta$ has two solutions, use analyticity to deduce that $F(\vect{z})$ is equal to the smallest solution and that $F(\vect{z}) <a$. 2.  When $\sum_k z_k \exp(k a) =\exp(a) -1$, go to $(tz_k)_{k\in \N}$ and exploit Abel's theorem for $t\nearrow 1$. The details are analogous to the continuous setting and we leave them  to the reader. 
\end{proof}

%%%%%%%%%%%%%%%%%%%%%%%%%%%%%%%%%%%%%%%%%%%%%%%%%%%%%%%%%%%%%%%%%%%%%%%%%%%%%

\section{Virial expansion}\label{sec:virial}

Here we prove Theorem~\ref{thm:vir} for the continuous system. The proof of Theorem~\ref{thm:vir-discrete} for the  discrete system is similar and therefore omitted.  A heuristic derivation of the relevant expressions for  both the continuous and the discrete systems is given in Section~\ref{sec:waals}. 

\begin{proof}[Proof of Theorem~\ref{thm:vir}]
 Fix $j \in \N$ and $z_k$, $k \neq j$. The inequality $\sum_k z_k \exp( - \ell_k \theta^*) >- \theta^*$ extends by continuity to some open neighborhood of $z_j$, and 
in this neighborhood the pressure solves the equation $F(z_j,p) = 0$ with $F(z_j,p) = p- \sum_k z_k \exp( - p \ell_k)$. Since 
\begin{equation} \label{eq:sec6F}
  \frac{\partial F}{\partial p}(z_j,p) = 1 + \sum_k z_k \ell_k e^{-p \ell_k}\geq 1>0, 
\end{equation} 
the implicit function theorem shows that the partial derivative $\frac{\partial p}{\partial z_j}$ exists and satisfies 
\begin{equation}  \label{eq:sec6a}
\rho_j(\boldsymbol{z})= z_j \frac{\partial p}{\partial z_j}(\boldsymbol{z})=  \frac{z_j \exp( - p(\boldsymbol{z}) \ell_j)}{ 1 + \sum_k z_k \ell_k \exp(-p(\boldsymbol{z}) \ell_k)}.
\end{equation} 
We multiply with $\ell_j$, sum over $j$, and get 
\begin{equation}  \label{eq:sec6pack}
  \sigma(\boldsymbol{z}) = \sum_j \ell_j \rho_j(\boldsymbol{z}) = \frac{\sum_jz_j \ell_j \exp( - p(\boldsymbol{z}) \ell_j)}{1+ \sum_j z_j \ell_j \exp( - p(\boldsymbol{z}) \ell_j)} <1. 
\end{equation} 
If we sum Eq.~\eqref{eq:sec6a} right away and use the fixed point equation for $p(\boldsymbol{z})$, we get 
\begin{equation*}
  \sum_j \rho_j(\boldsymbol{z}) = \frac{p(\boldsymbol{z})}{1+ \sum_k z_k \ell_k \exp( - p(\boldsymbol{z}) \ell_k)} = p(\boldsymbol{z})\bigl( 1- \sum_k \ell_k \rho_k(\boldsymbol{z}) \bigr), 
\end{equation*} 
which yields the expression of the pressure in terms of the densities. Eq.~\eqref{eq:sec6a} also shows 
\begin{equation*} 
  z_j = \rho_j(\boldsymbol{z}) e^{ p(\boldsymbol{z}) \ell_j}\bigl( 1+ \sum_k z_k \ell_k e^{- p(\boldsymbol{z}) \ell_k}\bigr) = \frac{ \rho_j(\boldsymbol{z})}{1- \sum_k \ell_k \rho_k(\boldsymbol{z})} e^{\ell_j p(\boldsymbol{z})}. \qedhere
\end{equation*} 
%This completes the proof for non-negative activities. 
%The argument for complex activities is similar. We leave the details to the reader and only check that the partial derivative $\partial F / \partial p$ is non-vanishing and that the packing fraction $\sigma(\boldsymbol{z})$ is different from $1$. Suppose that $\sum_k |z_k|\exp( a \ell_k) <a$ for some $a>0$. Let $\tilde z_k := -|z_k|$. Because of the alternating signs, we know that $|p( \boldsymbol{z})| \leq - p(\tilde {\boldsymbol{z}})=: F$, and by Eq.~\eqref{eq:Fbounds}, $\sum_k \ell_k \tilde z_k \exp( \ell_k F) <1$. Therefore 
%\begin{equation*} 
%  \bigl| 1 + \sum_k z_k \ell_k e^{-p \ell_k}\bigr| \geq 1 -  \sum_k {\tilde z}_k \ell_k e^{-F \ell_k}>0. 
%\end{equation*} 
%Therefore the partial derivative~\eqref{eq:sec6F} is non-vanishing and the packing fraction defined as in Eq.~\eqref{eq:sec6pack} is well-defined and different from $1$.
\end{proof}

%%%%%%%%%%%%%%%%%%%%%%%%%%%%%%%%%%%%55
\section{The inverse function theorem does not apply} \label{sec:inverse}

%Here we prove Theorem~\ref{thm:inverse}. 

\begin{proof} [Proof of Theorem~\ref{thm:inverse}] 
  Suppose by contradiction that $\rho(\cdot)$ is a bijection from $U_a$ onto $V_b$ with $U_a$ and $V_b$ neighborhoods of the origin in suitable spaces $E_a$, $E_b$. Let $\eps>0$, $\tilde \rho_k:= \eps k^{-3} \exp( - b k)$. We choose $\eps>0$ sufficiently small so that $C:=\sum_k k \rho_k <1$ and $\tilde {\vect{\rho}} = (\tilde \rho_k)_{k\in \N} \in V_b$. Since $\rho(\cdot)$ is assumed to be a bijection from $U_a$ onto $V_b$, there is a unique $\vect{z} \in U_a$ such that $\vect{\tilde \rho} = \vect{\rho}(\vect{z})$.  Theorem~\ref{thm:vir} shows that this activity vector is given by $z_k = (1-C)^{-1} \rho_k \exp( kp(\vect{z}))>0$, and we have $p(\vect{z})>0$. 
  Therefore 
\begin{equation*} 
  \frac{1}{1-C} \sum_k z_k e^{ak} = \frac{\eps}{1-C}\sum_k \frac{1}{k^3} e^{(p(\vect{z})+a-b)k} <\infty
\end{equation*} 
and $b \geq p+a >a$. Next set $z_k = - \eps k^{-3} \exp( - a k)$, with $\eps>0$ small enough so that $\vect{z} \in U_a$. Then $p(\vect{z})<0$ and  $\rho_k(\vect{z})$ is proportional to $z_k \exp(- k p(\vect{z}) )= \eps k^{-3} \exp( -(a+ p(\vect{z})) k)$. 
Since $\vect{\rho}(\vect{z}) \in V_b$ we must have $b\leq a+ p(\vect{z}) <a$. Thus $a<b$ and $b<a$, contradiction. 
\end{proof}

\textbf{Acknowledgments.} {I am indebted to R. Fern{\'a}ndez, S. J. Tate, D. Tsagkarogiannis and D. Ueltschi for many helpful discussions, and to A. van Enter for pointing out connections with the Fisher-Felderhof clusters.  This work was completed during a stay at the Institute for Computational and Experimental Research in Mathematics (ICERM), Brown University, for the semester program ``Phase Transitions and Emergent Properties.''
}

%%%%%%%%%%%%%%%%%%%%%%%%%%%%%%%%%%%%5
%\nocite{*}
\bibliography{tonks}

\newcommand{\etalchar}[1]{$^{#1}$}
\begin{thebibliography}{KLM{\etalchar{+}}02}

\bibitem[Ber08]{bernardi08}
O.~Bernardi.
\newblock Solution to a combinatorial puzzle arising from {M}ayer's theory of
  cluster integrals.
\newblock {\em S\'em. Lothar. Combin.}, 59:Art. B59e, 10 p., 2008.

\bibitem[BFP10]{bissacot-fernandez-procacci}
R.~Bissacot, R.~Fern{\'a}ndez, and A.~Procacci.
\newblock On the convergence of cluster expansions for polymer gases.
\newblock {\em J. Stat. Phys.}, 139:598--617, 2010.

\bibitem[BI03a]{brydges-imbrie}
D.~C. Brydges and J.~Z. Imbrie.
\newblock Branched polymers and dimensional reduction.
\newblock {\em Ann. of Math. (2)}, 158:1019--1039, 2003.

\bibitem[BI03b]{brydges-imbrie03b}
D.~C. Brydges and J.~Z. Imbrie.
\newblock Dimensional reduction formulas for branched polymer correlation
  functions.
\newblock {\em J. Statist. Phys.}, 110:503--518, 2003.

\bibitem[BM14]{brydges-marchetti14}
D.~Brydges and D.~H.~U. Marchetti.
\newblock On the virial series for a gas of particles with uniformly repulsive
  pairwise interaction, 2014.
\newblock Preprint, arXiv:1403.1621 [math-ph].

\bibitem[Bry11]{brydges-iamp}
D.~C. Brydges.
\newblock Mayer expansions.
\newblock {\em News Bulletin of the International Association of Mathematical
  Physics}, pages 11--15, April 2011.

\bibitem[DG13]{disertori-giuliani13}
M.~Disertori and A.~Giuliani.
\newblock The nematic phase of a system of long hard rods.
\newblock {\em Comm. Math. Phys.}, 323:143--175, 2013.

\bibitem[EH05]{evans-hanney05}
M.~R Evans and T.~Hanney.
\newblock Nonequilibrium statistical mechanics of the zero-range process and
  related models.
\newblock {\em Journal of Physics A: Mathematical and General}, 38(19):R195,
  2005.

\bibitem[EM94]{ehrenborg-mendez}
R.~Ehrenborg and M.~M{\'e}ndez.
\newblock A bijective proof of infinite variated {G}ood's inversion.
\newblock {\em Adv. Math.}, 103(2):221--259, 1994.

\bibitem[Far10]{faris10}
W.~G. Faris.
\newblock Combinatorics and cluster expansions.
\newblock {\em Probab. Surv.}, 7:157--206, 2010.

\bibitem[Fel71]{fellervol2}
W.~Feller.
\newblock {\em An introduction to probability theory and its applications.
  {V}ol. {II}.}
\newblock Second edition. John Wiley \& Sons, Inc., New York-London-Sydney,
  1971.

\bibitem[FF70]{fisher-felderhof70}
M.~E. Fisher and B.~U. Felderhof.
\newblock Phase transitions in one-dimensional cluster-interaction fluids.
  i{A}. {T}hermodynamics.
\newblock {\em Ann. Phys.}, 58:176--216, 1970.

\bibitem[FFG01]{fernandez-ferrari-garcia}
R.~Fern{\'a}ndez, P.~A. Ferrari, and N.~Garcia.
\newblock Loss network representation of {P}eierls contours.
\newblock {\em Ann. Probab.}, 29:902--937, 2001.

\bibitem[FG02]{fritzsche-grauert-book}
K.~Fritzsche and H.~Grauert.
\newblock {\em From holomorphic functions to complex manifolds}, volume 213 of
  {\em Graduate Texts in Mathematics}.
\newblock Springer-Verlag, New York, 2002.

\bibitem[Fis72]{fisher72}
M.~E. Fisher.
\newblock On discontinuity of the pressure.
\newblock {\em Comm. Math. Phys.}, 26, 1972.

\bibitem[FP07]{fernandez-procacci}
R.~Fern{\'a}ndez and A.~Procacci.
\newblock Cluster expansion for abstract polymer models. {N}ew bounds from an
  old approach.
\newblock {\em Comm. Math. Phys.}, 274:123--140, 2007.

\bibitem[FPS07]{fernandez-procacci-scoppola}
R.~Fern{\'a}ndez, A.~Procacci, and B.~Scoppola.
\newblock The analyticity region of the hard sphere gas. {I}mproved bounds.
\newblock {\em J. Stat. Phys.}, 128:1139--1143, 2007.

\bibitem[GK71]{gruber-kunz71}
C.~Gruber and H.~Kunz.
\newblock General properties of polymer systems.
\newblock {\em Comm. Math. Phys.}, 22:133--161, 1971.

\bibitem[Goo60]{good60}
I.~J. Good.
\newblock Generalizations to several variables of {L}agrange's expansion, with
  applications to stochastic processes.
\newblock {\em Proc. Cambridge Philos. Soc.}, 56:367--380, 1960.

\bibitem[Goo65]{good61}
I.~J. Good.
\newblock The generalization of {L}agrange's expansion and the enumeration of
  trees.
\newblock {\em Proc. Cambridge Philos. Soc.}, 61:499--517, 1965.

\bibitem[Haa04]{haas04}
B.~Haas.
\newblock Appearance of dust in fragmentations.
\newblock {\em Commun. Math. Sci.}, 2:65--73, 2004.

\bibitem[HP91]{hilton-pedersen}
P.~Hilton and J.~Pedersen.
\newblock Catalan numbers, their generalization, and their uses.
\newblock {\em Math. Intelligencer}, 13:64--75, 1991.

\bibitem[IVZ06]{ioffe-velenik-zahradnik06}
D.~Ioffe, Y.~Velenik, and M.~Zahradn{\'{\i}}k.
\newblock Entropy-driven phase transition in a polydisperse hard-rods lattice
  system.
\newblock {\em J. Stat. Phys.}, 122:761--786, 2006.

\bibitem[Joy88]{joyce88}
G.~S. Joyce.
\newblock On the hard-hexagon model and the theory of modular functions.
\newblock {\em Philos. Trans. Roy. Soc. London Ser. A}, 325:643--702, 1988.

\bibitem[JTTU14]{multivirial}
S.~Jansen, S.~J. Tate, D.~Tsagkarogiannis, and D.~Ueltschi.
\newblock Multispecies virial expansions.
\newblock {\em Comm. Math. Phys.}, 330:801--817, 2014.

\bibitem[KLM{\etalchar{+}}02]{kafri02}
Y.~Kafri, E.~Levine, D.~Mukamel, G.~M. Sch\"utz, and J.~T\"or\"ok.
\newblock Criterion for phase separation in one-dimensional driven systems.
\newblock {\em Phys. Rev. Lett.}, 89:035702, 2002.

\bibitem[LR64]{lebowitz-rowlinson64}
J.~L. Lebowitz and J.~S. Rowlinson.
\newblock Thermodynamic properties of mixtures of hard spheres.
\newblock {\em J. Chem. Phys.}, 41:133--138, 1964.

\bibitem[Moo70]{moon70}
J.~W. Moon.
\newblock {\em Counting labelled trees}, volume 1969 of {\em From lectures
  delivered to the Twelfth Biennial Seminar of the Canadian Mathematical
  Congress (Vancouver}.
\newblock Canadian Mathematical Congress, Montreal, Que., 1970.

\bibitem[PU09]{poghosyan-ueltschi09}
S.~Poghosyan and D.~Ueltschi.
\newblock Abstract cluster expansion with applications to statistical
  mechanical systems.
\newblock {\em J. Math. Phys.}, 50:053509, 17, 2009.

\bibitem[Rue69]{ruelle-book}
D.~Ruelle.
\newblock {\em Statistical mechanics: {R}igorous results}.
\newblock W. A. Benjamin, Inc., New York-Amsterdam, 1969.

\bibitem[Sok09]{sokal09}
Alan~D. Sokal.
\newblock A ridiculously simple and explicit implicit function theorem.
\newblock {\em S\'em. Lothar. Combin.}, 61A:Art. B61Ad, 2009.

\bibitem[Tat15]{tate14}
S.~J. Tate.
\newblock A solution to the combinatorial puzzle of {M}ayer's virial expansion.
\newblock {\em Ann. Inst. Henri Poincaré Comb. Phys. Interact.}, 2:229--262,
  2015.

\bibitem[Ton36]{tonks36}
L.~Tonks.
\newblock The complete equation of state of one, two and three-dimensional
  gases of hard elastic spheres.
\newblock {\em Phys. Rev.}, 50:955--963, 1936.

\end{thebibliography}
\bibliographystyle{alpha}

\end{document}